\documentclass[12pt,journal,onecolumn]{IEEEtran}

\usepackage{tikz}
\usepackage{fullpage}
\usepackage{amsfonts}
\usepackage{amsmath,amsbsy,pifont,subfigure,verbatim,amssymb,enumerate}

\usepackage{setspace}
\usepackage{amsthm}

 \usepackage{microtype}
 \usepackage{vwcol} 
 \usepackage{subfigure} 
\usepackage{amsmath}

\DeclareMathOperator*{\argmin}{argmin}

\usepackage[english]{babel} 

  \newtheorem{theorem}{Theorem}
  
 \newtheorem{remark}{Remark}
 
  \newtheorem{definition}{Definition}
 \newtheorem{lemma}[theorem]{Lemma} 

   \usepackage{subfig}
      \makeatletter
      \renewcommand{\paragraph}{\@startsection{paragraph}{4}{\z@}%
        {-3.25ex\@plus -1ex \@minus -.2ex}%
        {1.5ex \@plus .2ex}%
        {\normalfont\normalsize}}
      \makeatother
      
       \usepackage{cite}
\begin{document}

\title{Revisited Design Criteria For STBCs With Reduced Complexity ML Decoding}

\author{Asma~Mejri, Mohamed-Achraf~Khsiba,~\IEEEmembership{Student Member,~IEEE,}~and
        Ghaya~Rekaya-Ben Othman,~\IEEEmembership{Member,~IEEE}
\thanks{This paper is submitted for publication to IEEE Transactions on Wireless Communications. Part of the material in this paper has been accepted for publication in the proceedings of the Twelfth International Symposium on Wireless Communication Systems (ISWCS 2015), Brussels, Belgium, August 25-28, 2015.  Asma Mejri, Mohamed-Achraf Khsiba and Ghaya Rekaya-Ben Othman are with the Communications and Electronics Department
of Telecom-ParisTech, Paris, 75013, France. Emails: amejri,khsiba,rekaya@telecom-paristech.fr.}}

\maketitle
%\markboth{Submitted paper to IEEE Transactions on Communications}%

\begin{abstract}
The design of linear STBCs offering a low-complexity ML decoding using the well known Sphere Decoder (SD) has been extensively studied in last years. The first considered approach to derive design criteria for the construction of such codes is based on the Hurwitz-Radon (HR) Theory for mutual orthogonality between the weight matrices defining the linear code. This appproach served to construct new families of codes admitting fast sphere decoding such as \textit{multi-group decodable}, \textit{fast decodable}, and \textit{fast-group decodable} codes. In a second Quadratic Form approach, the Fast Sphere Decoding (FSD) complexity of linear STBCs is captured by a Hurwitz Radon Quadratic Form (HRQF) matrix based in its essence on the HR Theory. In this work, we revisit the structure of weight matrices for STBCs to admit Fast Sphere decoding. We first propose novel sufficient conditions and design criteria for reduced-complexity ML decodable linear STBCs considering an arbitrary number of antennas and linear STBCs of an arbitrary coding rate. Then we apply the derived criteria to the three families of codes mentioned above and provide analytical proofs showing that the FSD complexity depends only on the weight matrices and their ordering and not on the channel gains or the number of antennas and explain why the so far used HR theory-based approaches are suboptimal. 
%FSD of Cyclic Division Algebra-based Perfect Codes is studied as an example.
\end{abstract}

\section{Introduction and Preliminaries}
We consider in this work transmission over a coherent block-fading MIMO channel using $n_t$ transmit and $n_r$ receive antennas and coded modulations using length-$T$ linear Space-Time Block Codes. The complex-valued channel output is written as:
 \begin{equation}\label{sysc}
 \mathbf{Y} = \mathbf{H}\mathbf{X} + \mathbf{Z}
 \end{equation}
where $\mathbf{X} \in \mathbb{C}^{n_t \times T}$ is the codeword matrix sent over $T$ channel uses and belonging to a codebook $\mathcal{C}$. $\mathbf{Z} \in \mathbb{C}^{n_r \times T}$ represents a complex-valued AWGN of i.i.d. entries of variance $N_0$ per real-valued dimension. The channel fadings are represented by the matrix $\mathbf{H} \in \mathbb{C}^{n_r \times n_t}$. As coherent transmission is considered, the channel matrix $\mathbf{H}$ is assumed to be perfectly known (estimated) at the receiver. In addition, the fadings $h_{ij}$ are assumed to be complex circularly symmetric Gaussian random variables of zero-mean and unit variance.

As linear STBCs are concerned within this work, the used STBC encodes $\kappa$ complex information symbols represented by the complex-valued symbols vector $\mathbf{s}=\left[ s_1,...,s_{\kappa} \right]^{t}$ and the codeword matrix admits a Linear Dispersion (LD) form according to:
\begin{equation}
\mathbf{X} = \displaystyle\sum_{i=1}^{\kappa} \left( \Re{(s_i)}\mathbf{A}_{2i-1}+ \Im{(s_i)}\mathbf{A}_{2i} \right)
\end{equation}
where $\Re{(s_i)}$ and $\Im{(s_i)}$ correspond respectively to the real and imaginary parts of the $\kappa$ complex information symbols and matrices $\mathbf{A}_l, l=1,...,2\kappa$ are fixed $n_t \times T$ complex linearly independent matrices defining the code, known as LD or weight matrices. The rate of such codes is equal to $\frac{\kappa}{T}$ complex symbols per channel realization. When full rate codes are used, $\kappa=n_tT$. Moreover, we consider in this work $2^{q}-$QAM constellations with $q$ bits per symbol and for which the real and imaginary parts of the information symbols belong to a PAM modulation taking values in the set $\left[ -(q-1),(q-1) \right]$. (Similar results can be derived for hexagonal constellations).

In this work we are interested in the decoding of linear STBCs using Maximum Likelihood criterion. Accordingly, the receiver seeks an estimate $\hat{\mathbf{X}}$ of the transmitted codeword $\mathbf{X}$ by solving the minimization problem given by:
\begin{equation}
\hat{\mathbf{X}} = \argmin_{\mathbf{X} \in \mathcal{C}} \parallel \mathbf{Y}-\mathbf{H}\mathbf{X} \parallel^{2}
\end{equation}
ML decoding remains thus to find the codeword matrix that minimizes the squared norm $m(\mathbf{X})=\parallel \mathbf{Y}-\mathbf{H}\mathbf{X} \parallel^{2}$. The complexity of ML decoding is determined by the minimum number of values of $m(\mathbf{X})$ that needs to be computed to find the ML solution. It is upper bounded by $2^{qn_tT}$, the complexity of the exhaustive search. One way to avoid the high complexity of the exhaustive search consists in applying sequential decoding such as the Sphere Decoder \cite{Viterbo99}. We are interested in this work in linear STBCs that admit low-complexity ML decoding.

Constructions of such codes date back to the \textit{Complex Orthogonal} designs with the Alamouti code \cite{Alamouti98} and subsequent codes proposed in \cite{Tarokh99,Tirkkonen02}. This family of codes offers the least ML decoding complexity that is linear as function of the constellation size. Their main drawback is their low maximum rate. \textit{Quasi-Orthogonal} codes with full diversity and larger rates than the orthogonal designs have been later on proposed \cite{Jafarkhani01,Sharma03,Su04,Yuen05,Wang05,Khan06}. Recently, 3 main families of codes admitting low-complexity ML decoding have been discovered: \textit{multi-group decodable} \cite{Dao08,5719284,Karmakar09,Karmakar09bis,Rajan10}, \textit{fast decodable} \cite{Biglieri08,Srinath09,Paredes08,Sinnokrot10,Oggier10,Vehkalahti10,Luzzi11,Ren11,Vehkalahti12,Markin13}, and \textit{fast group decodable} codes \cite{Ren10}. A sub-class of fast decodable codes, termed \textit{Block-orthogonal codes} has been proposed in \cite{Ren11bis}. Information symbols in codes that belong to these families can be grouped into different partitions and decoded separately resulting in low-decoding complexity.  

The construction and study of the above mentioned families of low-complexity ML decoding codes has been performed based on the so-called \textit{Hurwitz-Radon Theory} (HR) to derive sufficient design criteria and conditions on the mutual orthogonality between the weight matrices defining the linear code. This theory has been later on used, in recent works, to define a second Quadratic Form approach in \cite{Jithamithra10,Jithamithra11,Jithamithra13}. The FSD complexity of linear STBCs is captured, under this approach, by a Hurwitz Radon Quadratic Form (HRQF) matrix. It is shown in \cite{Jithamithra11,Jithamithra13} that the Quadratic Form approach allows to determine the FSD complexity of the codes that belong to the families of multi-group decodable, fast decodable and fast-group decodable. Nevertheless, as highlighted in \cite{Jithamithra13}, it does not capture the sub-class of block-orthogonal codes. In this work, we revisit the design of the weight matrices for STBCs to admit low-complexity ML fast sphere decoding. The contributions of this work are as follows:
\begin{itemize}
\item We propose novel sufficient conditions and design criteria for reduced-complexity ML decodable linear STBCs considering an arbitrary number of antennas and linear STBCs of an arbitrary coding rate.
\item We apply the derived criteria to the families of multi-group decodable, fast decodable and fast group decodable codes and provide analytical proofs showing that the FSD complexity depends only on the weight matrices and their ordering and not on the channel gains or the number of antennas.
\item We provide analytical proofs showing that the HR Theory and HRQF-based approaches are suboptimal and explain why the latter does not allow to capture exactly the FSD complexity of all classes of STBCs and show that our derived design criteria are sufficient to construct Block Orthogonal codes.
%\item We study the FSD complexity of codes built using Cyclic Division Algebra using the derived design criteria.
\end{itemize}  
The remaining of this work is organized as follows: in Section II and III we introduce the system model and review the formal definitions of the main classes of low-complexity ML decoding codes essentially multi-group decodable, fast decodable and fast-group decabable codes. In Section IV, we derive novel sufficient design criteria for FSD of STBCs, apply them to the above mentioned 3 families of codes including the sub-class of block orthogonal codes and show analytically the suboptimality of the sufficient design conditions existing in literature based on HR theory. Results are summarized in a Section V.
%\vspace{-2cm}
\paragraph*{\textbf{Notation:} in this work we use the following notations: boldface letters are used for column vectors and capital boldface letters for matrices. Superscripts $^{t},~^{H}$ and $^{*}$ denote transposition, Hermitian transposition, and complex conjugation, respectively. $\mathbb{Z}$ and $\mathbb{C}$ denote respectively the ring of rational integers and the field of complex numbers. $i$ is the complex number such that $i^{2}=-1$. In addition, $\mathbf{I}_{n}$ denotes the $n \times n$ identity matrix. Furthermore, for a complex number $x$, we define the $(\tilde{.})$ operator from $\mathbb{C}$ to $\mathbb{R}^{2}$ as $\tilde{x}= \left[ \Re{(x)}, \Im{(x)} \right]^{t}$ where $\Re{(.)}$ and $\Im{(.)}$ denote real and imaginary parts. This operator can also be extended to a complex vector $\mathbf{x}=\left[ x_1,...,x_n \right]^{t} \in \mathbb{C}^{n}$ according to: $\tilde{\mathbf{x}} = \left[ \Re{(x_1)}, \Im{(x_1)}, ..., \Re{(x_n)}, \Im{(x_n)} \right]$. We define additionally the operator $\bar{\bar{(.)}}$ from $\mathbb{C}$ to $\mathbb{R}^{2}$ as $\bar{\bar{{x}}}= \left[-\Im{(x)},\Re{(x)} \right]^{t}$ Also, we define the operator $\check{(.)}$ from $\mathbb{C}$ to $\mathbb{R}^{2 \times 2}$ as:
\begin{equation}
\check{x} \stackrel{\triangle}{=} \left[ \begin{array}{cc}
\Re{(x)} & -\Im{(x)} \\ \Im{(x)} & \Re{(x)}
\end{array} \right] \notag
\end{equation}
The operator $\check{(.)}$ can be in a similar way extended to $n \times n$ matrices by applying it to all the entries of the matrix which results in a $2n \times 2n$ real-valued matrix. We define also the $\mathrm{vec}(.)$ operator that stacks the $m$ columns of an $n \times m$ complex-valued matrix into an $mn$ complex column vector. The $\parallel .\parallel$ operator denotes the Euclidean norm of a vector. We define also, for a complex number $x \in \mathbb{C}$ such that $x= \Re{(x)} + i \Im{(x)}$ the trace form such that $\mathrm{Tr}(x)=\mathrm{Tr}_{\mathbb{Q}(i)/\mathbb{Q}}(x)=2\Re{(x)}$.}
 
\section{System Model and real-valued vectorization of channel output}
%\subsection{ML decoding problem}
%ML decoding using sequential decoders in general and the Sphere Decoder in particular exploits the triangular structure of the ML metric. 
Without loss of generality, we consider in the remaining of this work sphere decoding. Results hold for any sequential decoders. In order to implement such decoders to solve the ML decoding problem, the complex-valued system in Eq.(\ref{sysc}) is transformed into a real-valued one using the vectorization operator $\mathrm{vec}(.)$ and the complex-to-real transformations $(\tilde{.})$ and $\check{(.)}$. We obtain accordingly:
\begin{equation}
\tilde{\mathrm{vec}}\left( \mathbf{Y}\right) = \mathbf{H}_{eq} \tilde{\mathbf{s}} + \tilde{\mathrm{vec}}\left( \mathbf{Z} \right)
\end{equation}
where $\mathbf{H}_{eq} \in \mathbb{R}^{2n_rT \times 2\kappa}$ is given by $\mathbf{H}_{eq}=\left( \mathbf{I}_{T} \otimes \check{\mathbf{H}} \right) \mathbf{G}$. The real-valued matrix $\mathbf{G} \in \mathbb{R}^{2n_tT \times 2\kappa}$, termed a generator matrix of the linear code satisfies $\tilde{\mathrm{vec}}(\mathbf{X}) = \mathbf{G} \tilde{\mathbf{s}}$ and can be written as function of the weight matrices as:
\begin{equation}
\mathbf{G} =  \left[ \tilde{\mathrm{vec}}\left( \mathbf{A}_1 \right) | \tilde{\mathrm{vec}}\left( \mathbf{A}_2 \right) | ... | \tilde{\mathrm{vec}}\left( \mathbf{A}_{2\kappa} \right) \right]
\end{equation} 
Given that the ordering of the weight matrices in the LD form corresponds to the order of the information symbols as $\Re{(s_1)},\Im{(s_1)},...,\Re{(s_{\kappa})},\Im{(s_{\kappa})}$ which corresponds to the considered order in the complex-to-real transformation using the operator $\tilde{(.)}$, any change of the ordering of the information symbols results in a similar modification in the ordering of the weight matrices.
%Any changing in the ordering of the symbols in the vector $\tilde{\mathbf{s}}$ changes the order of the weight matrices in the generator matrix $\mathbf{G}$.
 The obtained real system can then be written in the form:
\begin{equation}\label{mlpb}
\mathbf{y} = \mathbf{H}_{eq}\tilde{\mathbf{s}} + \mathbf{z}
\end{equation}
Using this equivalent system, the ML decoding metric is equivalently written by:
\begin{equation}
m(\tilde{\mathbf{s}}) = \parallel \mathbf{y}-\mathbf{H}_{eq}\tilde{\mathbf{s}} \parallel^{2} = \parallel \mathbf{Q}^{t}\mathbf{y} - \mathbf{R}\tilde{\mathbf{s}} \parallel^{2} 
\end{equation}
Where $\mathbf{Q} \in \mathbb{R}^{2n_rT \times 2\kappa}$ is an orthogonal matrix and $\mathbf{R} \in \mathbb{R}^{2\kappa \times 2\kappa}$ is upper triangular obtained both from the QR decomposition of the equivalent channel matrix $\mathbf{H}_{eq}=\mathbf{Q}\mathbf{R}$. Using Gram-Schmidt orthogonolization, matrices $\mathbf{Q}$ and $\mathbf{R}$ are given by: $\mathbf{Q}= \left[ \mathbf{q}_1 | \mathbf{q}_2 |... | \mathbf{q}_{2\kappa} \right]$ where $\mathbf{q}_i, i=1,...,2\kappa$ are column vectors and:
\begin{align}
\mathbf{R} &= \left[ \begin{array}{ccccc}
\parallel \mathbf{r}_1 \parallel & <\mathbf{q}_1,\mathbf{h}^{eq}_2> & <\mathbf{q}_1,\mathbf{h}^{eq}_3> & \cdots & <\mathbf{q}_1,\mathbf{h}^{eq}_{2\kappa}> \\
0 & \parallel \mathbf{r}_2 \parallel & <\mathbf{q}_2,\mathbf{h}^{eq}_3> & \cdots & <\mathbf{q}_2,\mathbf{h}^{eq}_{2\kappa}> \\
0 & 0 & \parallel \mathbf{r}_3 \parallel & \cdots & <\mathbf{q}_3,\mathbf{h}^{eq}_{2\kappa}> \\
\vdots & \vdots & \vdots & \ddots & \cdots \\
0 & 0 & 0 & \cdots & \parallel \mathbf{r}_{2\kappa} \parallel
\end{array} \right]
\end{align}
 where $\mathbf{r}_1=\mathbf{h}^{eq}_1, \mathbf{q}_1=\frac{\mathbf{r}_1}{\parallel \mathbf{r}_1\parallel}$ and for $i=2,...,2\kappa$, $\mathbf{r}_i = \mathbf{h}^{eq}_i - \sum_{j=1}^{i-1} <\mathbf{q}_j,\mathbf{h}^{eq}_i> \mathbf{q}_j ~,~ \mathbf{q}_{i}=\frac{\mathbf{r}_i}{\parallel \mathbf{r}_i \parallel}$.
\begin{remark}
The provided results in this work are based on the complex-to-real tranformations $\tilde{(.)}$ and $\check{(.)}$ and the column-wize vectorization operation and show that the zero structure of the matrix $\mathbf{R}$ depends only on the weight matrices and their ordering in the matrix $\mathbf{G}$ which corresponds to the ordering of the real and imaginary parts of the complex information symbols in the vector $\tilde{\mathbf{s}}$. Any changing in this ordering impacts the number and locations of zero entries in the matrix $\mathbf{R}$. This can be observed as follows. Let $\mathbf{P}_s$ be a permutation matrix that allows to change the ordering of the elements in $\tilde{\mathbf{s}}$ applied to $\tilde{\mathbf{s}}$ and similarly $\mathbf{P}_y$ a permutation matrix applied to $\tilde{\mathrm{vec}}\left( \mathbf{Y}\right)$ and $\tilde{\mathrm{vec}}\left( \mathbf{Z}\right)$. Given that the permutation matrix is orthogonal, we obtain the equivalent system given by:
\vspace{-0.5cm}
\begin{align}
\mathbf{P}_y\tilde{\mathrm{vec}}\left( \mathbf{Y}\right) &=\mathbf{P}_y\mathbf{H}_{eq} \tilde{\mathbf{s}} + \mathbf{P}_y\tilde{\mathrm{vec}}\left( \mathbf{Z} \right) = \mathbf{P}_y\left( \mathbf{I}_{T} \otimes \check{\mathbf{H}} \right)\mathbf{G} \tilde{\mathbf{s}} + \mathbf{P}_y\tilde{\mathrm{vec}}\left( \mathbf{Z} \right) \notag \\
&=  \mathbf{P}_y\left( \mathbf{I}_{T} \otimes \check{\mathbf{H}} \right)\mathbf{G} \mathbf{P}_s^{t} \mathbf{P}_s \tilde{\mathbf{s}} + \mathbf{P}_y\tilde{\mathrm{vec}}\left( \mathbf{Z} \right)= \mathbf{H}_{eq,ord} \tilde{\mathbf{s}}_{ord} + \mathbf{P}_y\tilde{\mathrm{vec}}\left( \mathbf{Z} \right) \notag
\end{align}
where $\mathbf{H}_{eq,ord}=\mathbf{P}_y\mathbf{H}_{eq}\mathbf{P}_s^{t}$ and $\tilde{\mathbf{s}}_{ord}=\mathbf{P}_s\tilde{\mathbf{s}}$. Then, let $\mathbf{H}_{eq}=\mathbf{Q}\mathbf{R}$ and $\mathbf{H}_{eq,ord}=\mathbf{Q}_1\mathbf{R}_1$ the QR decompositions of the equqivalent channel matrix before and after permutation. We can write:
\begin{align}
\mathbf{H}_{eq,ord}=\mathbf{Q}_1\mathbf{R}_1 & \Leftrightarrow \mathbf{P}_y\mathbf{H}_{eq}\mathbf{P}_s^{t}=\mathbf{Q}_1\mathbf{R}_1 \Leftrightarrow \mathbf{H}_{eq}=\mathbf{P}_y^{t}\mathbf{Q}_1\mathbf{R}_1\mathbf{P}_s \Leftrightarrow \mathbf{Q}\mathbf{R}=\mathbf{P}_y^{t}\mathbf{Q}_1\mathbf{R}_1\mathbf{P}_s \notag
%&\Leftrightarrow \mathbf{Q}_1\mathbf{R}_1=\mathbf{P}\mathbf{Q}\mathbf{P}^{t}\mathbf{P}\mathbf{R}\mathbf{P}^{t} \\
%&\Leftrightarrow \mathbf{Q}_1=\mathbf{P}\mathbf{Q}\mathbf{P}^{t}, \mathbf{R}_1=\mathbf{P}\mathbf{R}\mathbf{P}^{t}
\end{align}
The relation between the matrix $\mathbf{R}$ of the QR decomposition of the equivalent channel matrix before and after applying the permutation shows that $\mathbf{P}_y$ does not impact the zero structure of the matrix $\mathbf{R}$ while any permutation of the symbols in the vector $\tilde{\mathbf{s}}$ and equivalently the columns of the matrix $\mathbf{G}$ has an impact on the locations and the numbers of the zero entries in the matrix $\mathbf{R}$. This observation sheds light on the investigation of the optimal permutation matrix $\mathbf{P}_s$ that results in a matrix $\mathbf{R}$ enabling ML decoding with the least complexity.
\end{remark}
When sequential decoding, for example using the Sphere Decoder, is used to solve this minimization problem, its complexity can be alleviated thanks to zero entries in the matrix $\mathbf{R}$, that depend on the used code and the ordering of the real and imaginary parts of the symbols in the vector $\tilde{\mathbf{s}}$ as discussed above.
 
\section{FSD Complexity and Main classes of Low-Complexity ML decoding codes}
In literature, we distinguish 3 main classifications of codes resulting in matrix $\mathbf{R}$ having an interesting form enabling reduced-complexity ML decoding: \textit{Multi-group decodable codes}, \textit{Fast decodable codes} and \textit{Fast group decodable codes}.  By structure, its meant the locations of the zero entries $R_{ij}$. The construction of such codes and determination of the structure of the matrix $\mathbf{R}$ for these classes is, commonly in literature, determined using a \textit{mutual orthogonality} property of the weight matrices based on which two main approaches are proposed in literature: \textit{Hurwitz-Radon} theory (HR) approach and a \textit{Quadratic Form} (HRQF) approach. We detail in the following subsections these two approaches and summarize the main results. We provide for convenience the definition of a partition as follows.
\begin{definition}[A set partition]
We call a partition $\left\lbrace a_1,...,a_n \right\rbrace$ into $g$ non-empty subsets $\Gamma_1,...,\Gamma_g$ with cardinalities $K_1,...,K_g$ an ordered partition if $\left\lbrace a_1,...,a_{K_1} \right\rbrace \in \Gamma_1$, $\left\lbrace a_{K_1+1},...,a_{K_1+K_2} \right\rbrace \in \Gamma_2$, so on till $\left\lbrace a_{\sum_{i=1}^{g-1}K_{i}+1},...,a_{\sum_{i=1}^{g}K_i} \right\rbrace \in \Gamma_g$.
\end{definition}

\subsection{Hurwitz-Radon theory-based approach}
The HR theory-based approach uses in its essence the mutual orthogonality of weight matrices \cite{Radon22}. The main result is stated in the following theorem \cite{Srinath09}.
\begin{theorem}\label{HRT}
For an STBC with $\kappa$ independent complex information symbols and $2\kappa$ linearly independent matrices $\mathbf{A}_l,l=1,...,2\kappa$, if, for any $i$ and $j$, $i \neq j$, $1 \leq i \neq j \leq 2\kappa$, $\mathbf{A}_i\mathbf{A}_j^{H}+\mathbf{A}_j\mathbf{A}_i^{H}=\mathbf{0}_{n_t}$, then the $i^{th}$ and $j^{th}$ columns of the equivalent channel matrix $\mathbf{H}_{eq}$ are orthogonal.
\end{theorem}
This property has been used to define and construct particular classes of codes defined below.
\begin{definition}[\emph{Multi-group decodable codes}]
An STBC is said to be $g-$group decodable if there exists a parition of $\left\lbrace 1,2,...,2\kappa \right\rbrace$ into $g$ non empty subsets $\Gamma_1, \Gamma_2,...,\Gamma_g$ such that $\mathbf{A}_l\mathbf{A}_m^{H}+\mathbf{A}_m\mathbf{A}_l^{H}=\mathbf{0}$, whenever $l \in \Gamma_i$ and $m \in \Gamma_j$ and $i \neq j$. The corresponding $\mathbf{R}$ matrix has the following form:
\begin{equation}
\mathbf{R} = \left[ \begin{array}{cccc}
\Delta_1 & \mathbf{0} & \cdots & \mathbf{0} \\
\mathbf{0} & \Delta_2 & \cdots & \mathbf{0} \\
\vdots & \vdots & \ddots & \vdots \\
\mathbf{0} & \mathbf{0} & \cdots & \Delta_g
\end{array} \right]
\end{equation}
where $\Delta_i, i=1,...,g$ is a square upper triangular matrix.
\end{definition}
\vspace{-0.5cm}
\paragraph*{\textbf{Example 1}: an example of multi-group decodable codes is the \textbf{ABBA code} proposed in \cite{Tirkkonen00} initially for $2 \times 2$ MIMO systems and afterwards generalized to configurations with more than $3$ transmit antennas. In the case of $n_t=n_r=2$, the ABBA code encodes $2$ complex information symbols or equivalently $4$ real-valued symbols $x_i \in \mathbb{R}, i=1,...,4$. The codeword matrix is accordingly written in the form:
\vspace{-0.5cm}
\begin{equation}\label{mat_abba}
\mathbf{X}_{ABBA} = \left[ \begin{array}{cc}
x_1 + i x_4 & -x_2+ix_3 \\
-x_2+ix_3 & x_1+ix_4
\end{array} \right]
\end{equation}
The corresponding matrix $\mathbf{R}$ matrix with the ordering of the real-valued symbols according to $\left[x_1,x_2,x_3,x_4 \right]$ is given by figure.\ref{abbafig}. 
\vspace{-0.5cm}
\begin{figure}[htp]
  \centering
\includegraphics[height=3cm,width=3cm]{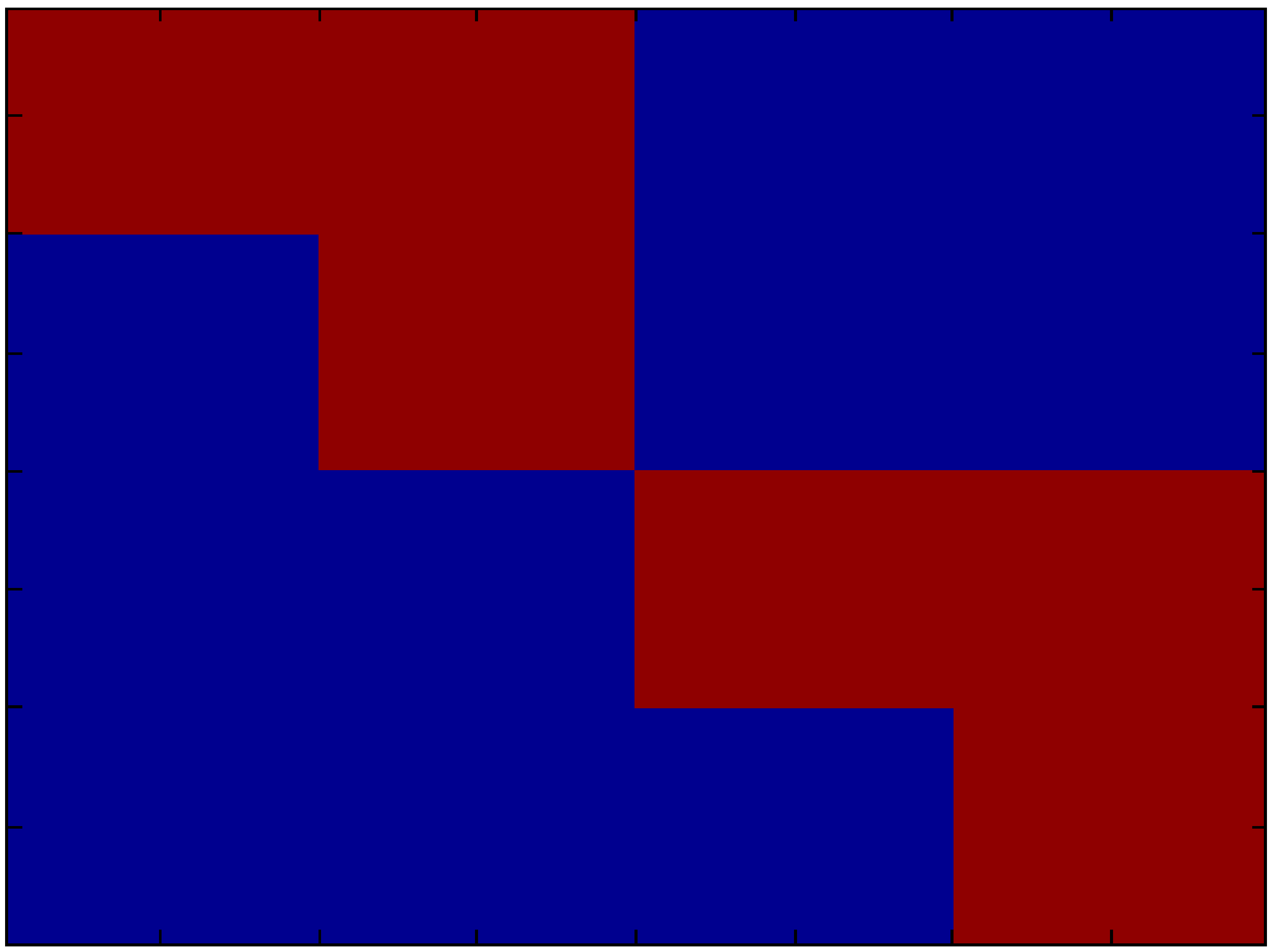}
  \caption{$\mathbf{R}$ matrix for the $2 \times 2$ ABBA code. Blue entries are zeros and red are arbitrary values.}\label{abbafig}
  \end{figure}}

\begin{definition}[\emph{Fast decodable codes}]
An STBC is said to be fast Sphere Decodable code if there exists a partition of $\left\lbrace 1,2,...,L \right\rbrace$ where $L \leq 2\kappa$ into $g$ non empty subsets $\Gamma_1, \Gamma_2,...,\Gamma_g$ such that $<\mathbf{q}_i,\mathbf{h}^{eq}_j>=0, (i < j)$, whenever $i \in \Gamma_p$ and $j \in \Gamma_q$ and $p \neq q$, where $\mathbf{q}_i$ and $\mathbf{h}^{eq}_j$ are column vectors of respectively $\mathbf{Q}$ and $\mathbf{H}_{eq}$. The corresponding $\mathbf{R}$ matrix has the following form:
\begin{equation}
\mathbf{R} = \left[ \begin{array}{cc}
\Delta & \mathbf{B}_1 \\
\mathbf{0} & \mathbf{B}_2 
\end{array} \right]
\end{equation}
where $\Delta$ is an $L \times L$ block diagonal, upper triangular matrix, $\mathbf{B}_1$ is a rectangular matrix and $\mathbf{B}_2$ is a square upper triangular matrix.
\end{definition}

\paragraph*{\textbf{Example 2}: an example of fast decodable codes is the famous \textbf{Silver Code} used in a $2 \times 2$ MIMO system and encoding $4$ complex symbols $s_1,...,s_4$. The codeword matrix is given by:
\vspace{-0.5cm}
\begin{equation}
\mathbf{X}_{sc} = \mathbf{X}_a\left(s_1,s_2\right)+\mathbf{X}_b\left(z_1,z_2\right)
\end{equation}
where $\mathbf{T}=\left[\begin{array}{cc}
1 & 0 \\ 0 & -1
\end{array} \right]$ is unimodular and $\mathbf{X}_a$ and $\mathbf{X}_b$ take the Alamouti structure as:
\vspace{-0.5cm}
\begin{equation}
\mathbf{X}_a \left(s_1,s_2\right) = \left[ \begin{array}{cc}
s_1 & -s^{*}_2 \\
s_2 & s^{*}_1
\end{array} \right] ~,~\mathbf{X}_b \left(z_1,z_2\right) = \left[ \begin{array}{cc}
z_1 & -z^{*}_2 \\
z_2 & z^{*}_1
\end{array} \right] \notag
\end{equation}
and
\vspace{-0.5cm}
\begin{equation}
\left[z_1,z_2\right]^{t} = \mathbf{U}\left[s_3,s_4 \right]^{t} ~,~ \mathbf{U}=\left[\begin{array}{cc}
\frac{1}{\sqrt{7}}(1+i) & \frac{1}{\sqrt{7}}(-1+2i) \\ \frac{1}{\sqrt{7}}(1+2i) & \frac{1}{\sqrt{7}}(1-i)
\end{array} \right] \notag
\end{equation}
The cordeword matrix is written as function of the complex information symbols as:
\begin{equation}
\mathbf{X}_{sc} = \left[ \begin{array}{cc}
s_1 + \frac{1}{\sqrt{7}}\left[(1+i)s_3+(-1+2i)s_4\right] & -s^{*}_2-\frac{1}{\sqrt{7}}\left[(1-2i)s^{*}_3+(1+i)s^{*}_4\right] \\
-s_2-\frac{1}{\sqrt{7}}\left[(1+2i)s_3+(1-i)s_4 \right] & s^{*}_1-\frac{1}{\sqrt{7}}\left[(1-i)s^{*}_3+(-1-2i)s^{*}_4\right]
\end{array} \right]
\end{equation}
The $\mathbf{R}$ matrix obtained when the symbols are ordered as $\left[ \Re{(s_1)},\Im{(s_1)},\Re{(s_2)},\Im{(s_2)},\Re{(s_3)},\Im{(s_3)},\Re{(s_4)},\Im{{s_4}} \right]$ has the form depicted in figure \ref{sc}.
\begin{figure}[htp]
  \centering
\includegraphics[height=3cm,width=3cm]{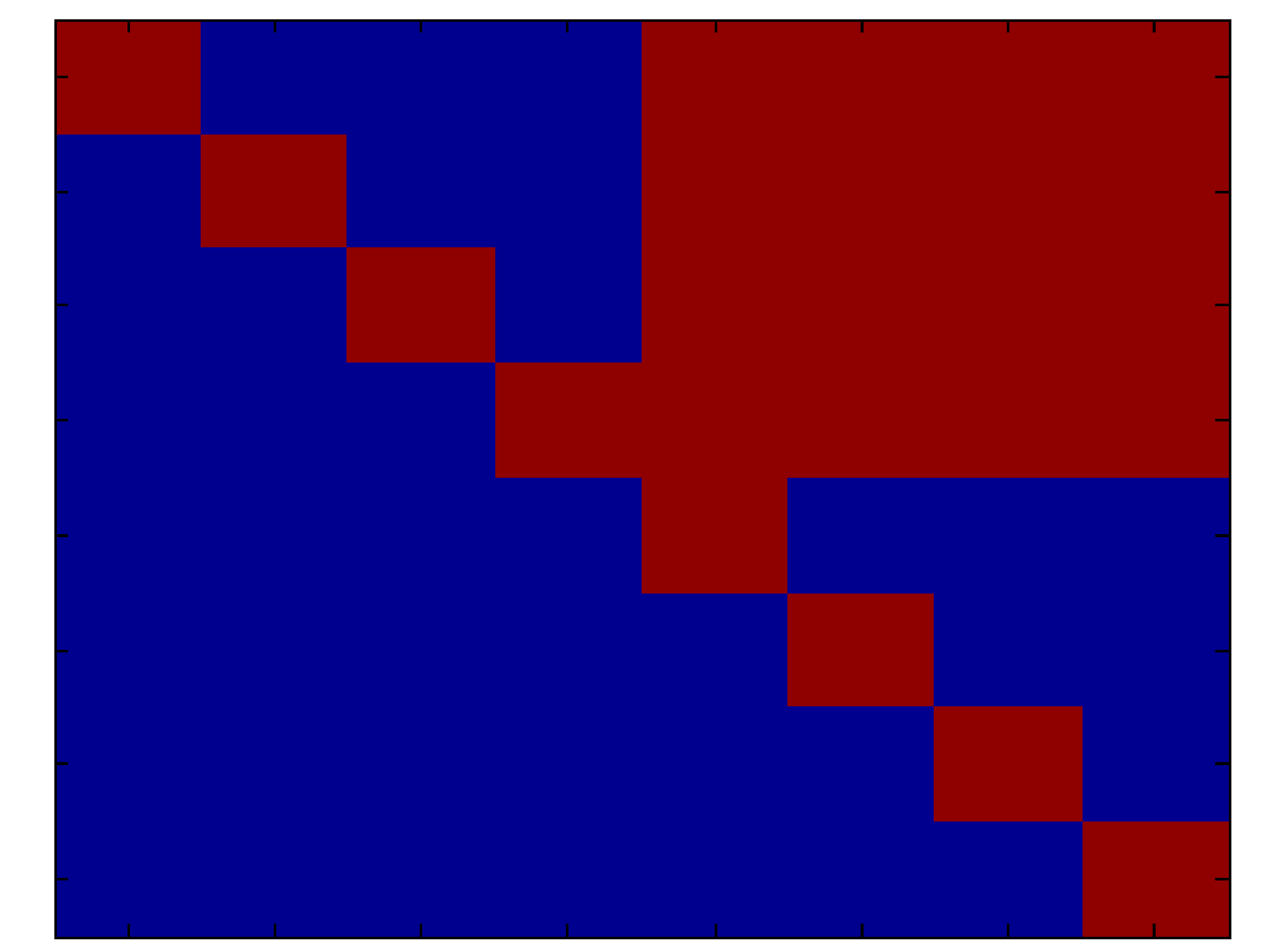}
  \caption{$\mathbf{R}$ matrix for the Silver code. Blue entries are zeros and red are arbitrary values.}\label{sc}
  \end{figure}
}

\begin{definition}[\emph{Fast-group decodable codes}]
An STBC with weight matrices $\mathbf{A}_l, l=1,...,2\kappa$ is said to be fast group decodable if it satisfies the following conditions:
\begin{itemize}
\item There exists a partition of $\left\lbrace 1,...,2\kappa \right\rbrace$ into $g$ non empty subsets $\Gamma_1,\Gamma_2,...,\Gamma_g$ such that $\mathbf{A}_l\mathbf{A}_m^{H}+\mathbf{A}_m\mathbf{A}_l^{H}=0$ whenever $l \in \Gamma_i$ and $m \in \Gamma_j$ and $i \neq j$.
\item In any of the partition $\Gamma_i$, we have $<\mathbf{q}_{i_{l_1}},\mathbf{h}^{eq}_{i_{l_2}}>=0$ $(l_1=1,2,...,L_i-1 ~\mathrm{and}~ l_2=l_1+1,...,L_i)$ and $L_i \leq |\Gamma_i|$ where $i=1,2,...,g$.
\end{itemize}
The corresponding $\mathbf{R}$ matrix has the following form:
\begin{equation}
\mathbf{R} = \left[ \begin{array}{cccc}
\mathbf{R}_1 & \mathbf{0} & \cdots & \mathbf{0} \\
\mathbf{0} & \mathbf{R}_2 & \cdots & \mathbf{0} \\
\vdots & \vdots & \ddots & \vdots \\
\mathbf{0} & \mathbf{0} & \cdots & \mathbf{R}_g
\end{array} \right]
\end{equation}
where at least one $\mathbf{R}_i, i=1,...,g$ has the fast-decodability form as: $\mathbf{R}_i = \left[ \begin{array}{cc}
\Delta_i & \mathbf{B}_{i_1} \\
\mathbf{0} & \mathbf{B}_{i_2} 
\end{array} \right]$, with $\Delta_i$ is an $L_i \times L_i$ block diagonal, upper triangular matrix, $\mathbf{B}_{i_1}$ is a rectangular matrix and $\mathbf{B}_{i_2}$ is a square upper triangular matrix.
\end{definition}
\vspace{-0.5cm}
\paragraph*{\textbf{Example 3}: an example of fast-group decodable codes is the code of proposed in \cite{Ren10} for a $4 \times 4$ system based on an orthogonal STBC with code rate $\frac{3}{4}$ symbols/c.u.\\
} 
\vspace{-0.5cm}
In addition to these families of codes, recently, a particular sub-class of fast decodable codes has been proposed, termed \textit{Block-Orthogonal codes}. Several known codes in literature belong to this family of codes, such as the BHV \cite{Biglieri08} code, the Silver code \cite{Hollanti08}, Srinath-Rajan code \cite{Srinath09}, codes from Cyclic Division algebras \cite{Sethuraman03}, crossed product algebras \cite{Shashidhar06} and fast decodable asymmetric STBCs from division algebras \cite{Vehkalahti12}. Block Orthogonal codes are fast decodable codes that depict additional structural conditions on the $\mathbf{R}$ matrix that has the form:
\begin{equation}
\mathbf{R} = \left[ \begin{array}{cccc}
\mathbf{R}_1 & \mathbf{B}_{12} & \cdots & \mathbf{B}_{1\Gamma} \\
\mathbf{0} & \mathbf{R}_2 & \cdots & \mathbf{B}_{2\Gamma} \\
\vdots & \vdots & \ddots & \vdots \\
\mathbf{0} & \mathbf{0} & \cdots & \mathbf{R}_{\Gamma}
\end{array} \right]
\end{equation}
where each matrix $\mathbf{R}_i,i=1,...,\Gamma$ is full rank, block diagonal, upper triangular with $k$ blocks $\mathbf{U}_{i1},...,\mathbf{U}_{ik}$ each of size $\gamma \times \gamma$ and $\mathbf{B}_{ij}, i=1,...,\Gamma, j=i+1,...,\Gamma$ are non-zero matrices. Codes satisfying such properties are called $(\Gamma,k,\gamma)$ Block Orthogonal STBCs. The formal definition of the sufficient design criteria for Block-Orthogonal codes were first given in \cite{Ren11bis} for codes with parameters $(\Gamma,k,1)$ and recently generalized in \cite{Jithamithra13bis,Jithamithra14} for codes with parameters $(\Gamma,k,\gamma)$ for arbitrary sizes of sub-blocks considering the matrices $\mathbf{R}_i,i=1,...,\Gamma$ having the same size of sub-blocks. These sufficient design conditions are summarized in the following lemma \cite{Jithamithra13bis,Jithamithra14}.
\begin{definition}[\emph{Block Orthogonal codes}]\label{bo}
Let the $\mathbf{R}$ matrix of an STBC with weight matrices $\left\lbrace \mathbf{A}_1,...,\mathbf{A}_L \right\rbrace$ and $\left\lbrace \mathbf{B}_1,...,\mathbf{B}_l \right\rbrace$ be $\mathbf{R}=\left[\begin{array}{cc}
\mathbf{R}_1 & \mathbf{E} \\ \mathbf{0} & \mathbf{R}_2
\end{array} \right]$, where $\mathbf{R}_1$ is an $L \times L$ upper triangular block-orthogonal matrix with parameters $\left(\Gamma-1,k,\gamma \right)$, $\mathbf{E}$ is an $L \times l$ matrix and $\mathbf{R}_2$ is an $l \times l$ upper triangular matrix. The STBC will be block orthogonal with parameters $(\Gamma,k,\gamma)$ if the following conditions are satisfied:
\begin{itemize}
\item The matrices $\left\lbrace \mathbf{B}_1,...,\mathbf{B}_l \right\rbrace$ are $k-$group decodable with $\gamma$ variables in each group.
\item The matrices $\left\lbrace \mathbf{A}_1,...,\mathbf{A}_L \right\rbrace$ when used as weight matrices for an STBC yield an $\mathbf{R}$ having a block orthogonal structure with parameters $(\Gamma-1,k,\gamma)$. When $\Gamma=2$, then $L=l$ and the matrices $\left\lbrace \mathbf{A}_1,...,\mathbf{A}_L \right\rbrace$ are $k-$group decodable with variables $\gamma$ in each group.
\item The set of matrices  $\left\lbrace \mathbf{A}_1,...,\mathbf{A}_L, \mathbf{B}_1,...,\mathbf{B}_l \right\rbrace$ are such that the matrix $\mathbf{R}$ obtained is of full rank.
\item The matrix $\mathbf{E}^{t}\mathbf{E}$ is a block diagonal matrix with $k$ blocks of size $\gamma \times \gamma$.
\end{itemize}
\end{definition}

\paragraph*{\textbf{Example 4}: an example of Block-Orthogonal codes is the \textbf{Golden code}. It employs 2 transmit and 2 receive antennas and encodes $4$ complex QAM symbols over two time slots ($T=2$) achieving full rate and full diversity \cite{Belfiore05}. The Golden code codeword matrix has then the form:
\begin{align}
\mathbf{X} &= \frac{1}{\sqrt{5}} \left[ \begin{array}{cc}
\alpha (s_1+\theta s_2) & \alpha (s_3+\theta s_4) \\
i \bar{\alpha}(s_3+\bar{\theta}s_4) & \bar{\alpha}(s_1+\bar{\theta}s_2)
\end{array} \right]
\end{align}
where 
\begin{align}
\theta &= \frac{1+\sqrt{5}}{2},~\bar{\theta}=\frac{1-\sqrt{5}}{2},~\alpha=1+i-i\theta~,~\bar{\alpha} = 1 + i - i \bar{\theta}
\end{align}
and $s_i,i=1,...,4$, are the transmitted symbols taken from the $2^{q}-$QAM constellation. From an algebraic construction point of view, the Golden code codebook is built using the base field $F=\mathbb{Q}(i)$ and $K=\mathbb{Q}(i,\theta)$ as an extension of $\mathbb{Q}(i)$ of degree $2$. The generator of the Galois group of $K/\mathbb{Q}(i)$ is $\sigma$ such that $\sigma(\theta)=\bar{\theta}$. The integral basis $\mathit{B}_I=\left(\nu_1,\nu_2 \right)=\left(\alpha,\alpha \theta \right)$ and the generator matrix is given by: $\mathbf{M} = \left[ \begin{array}{cc}
\alpha & \alpha\theta \\
\sigma(\alpha) & \sigma(\alpha \theta)
\end{array} \right] = \left[ \begin{array}{cc}
\alpha & \alpha\theta \\
\bar{\alpha} & \bar{\alpha} \bar{\theta}
\end{array} \right]$. The Golden code is Block Orthogonal and its corresponding $\mathbf{R}$ matrix for an order of the real-valued symbols as $\left[ \Re{(s_1)},\Re{(s_2)},\Im{(s_1)},\Im{(s_2)},\Re{(s_3)},\Re{(s_4)},\Im{(s_3)},\Im{(s_4)} \right]$ is given by figure \ref{goldencmat}.
\begin{figure}[htp]
  \centering
\includegraphics[height=3cm,width=3cm]{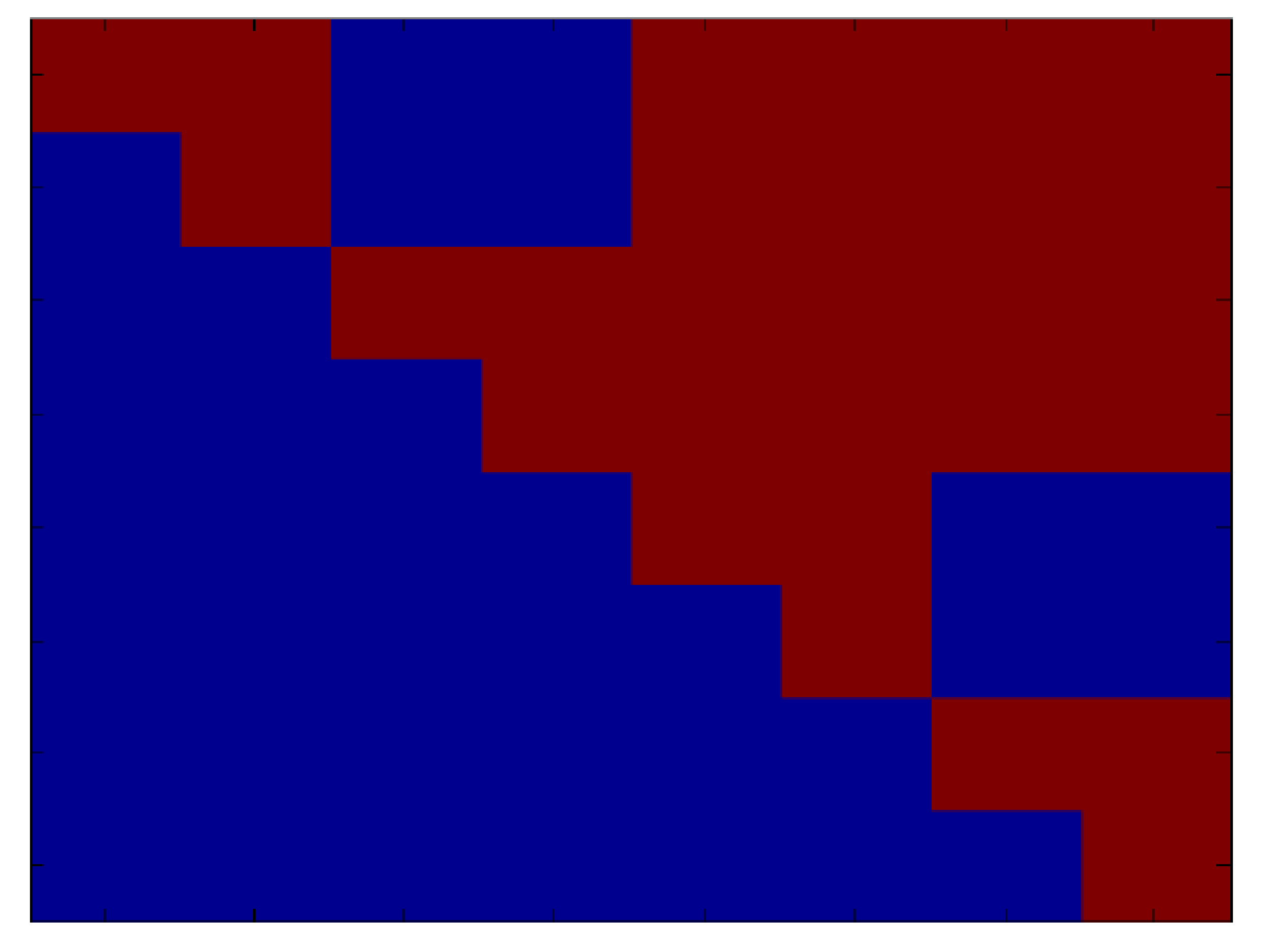}
  \caption{$\mathbf{R}$ matrix for the Golden code.}\label{goldencmat}
  \end{figure}}
\vspace{-0.5cm}
\subsection{Quadratic Form-based approach}
In theorem \ref{HRT}, it is shown that the Hurwitz-Radon Theory capturing the orthogonality between two weight matrices is sufficient to obtain orthogonality of corresponding columns of the equivalent channel matrix. This property was later on used in \cite{Jithamithra10,Jithamithra11,Jithamithra13} to develop a Quadratic Form termed \textit{Hurwitz Radon Quadratic Form} (HRQF). This quadratic form has been considered before in literature \cite{Unger11} to determine whether Quaternion algebras or Biquaternion algebras are division algebras. In \cite{Jithamithra10,Jithamithra11,Jithamithra13}, this quadratic form is further exploited to define the zero structure of the matrix $\mathbf{R}$ by associating to the HRQF a matrix $\mathbf{U}$ such that $U_{ij}= \parallel \mathbf{A}_i \mathbf{A}_j^{H}+\mathbf{A}_j\mathbf{A}_i^{H} \parallel^{2}$ and $U_{ij}=0$ if and only if $\mathbf{A}_i \mathbf{A}_j^{H}+\mathbf{A}_j\mathbf{A}_i^{H}=\mathbf{0}_{n_t}$. This form has been used in \cite{Jithamithra10,Jithamithra11,Jithamithra13} to determine sufficient conditions for an STBC to admit multi-group, fast and fast-group decodability as summarized in the following lemmas \cite{Jithamithra13}. 

\begin{lemma}[\emph{Multi-group decodability}]
Let an STBC with $\kappa$ independent complex symbols, $2\kappa$ weight matrices and HRQF matrix $\mathbf{U}$. If there exists an ordered partition of $\left\lbrace 1,2,...,2\kappa \right\rbrace$ into $g$ non empty subsets $\Gamma_1,...,\Gamma_g$ such that $U_{ij}=0$ whenever $i \in \Gamma_p$ and $j \in \Gamma_q$ and $p \neq q$, then the code is $g-$group sphere decodable.
\end{lemma}
  
\begin{lemma}[\emph{Fast decodability}]
Let an STBC with $\kappa$ independent complex symbols, $2\kappa$ weight matrices and HRQF matrix $\mathbf{U}$. If there exists a partition of $\left\lbrace 1,2,...,L \right\rbrace$ where $L \leq 2\kappa$ into $g$ non empty subsets $\Gamma_1, \Gamma_2,...,\Gamma_g$ such that $U_{ij}=0$ whenever $i \in \Gamma_p$ and $j \in \Gamma_q$ and $p \neq q$, then the code is fast decodable.
\end{lemma}

\begin{lemma}[\emph{Fast-group decodability}]
Let an STBC with $\kappa$ independent complex symbols, $2\kappa$ weight matrices and HRQF matrix $\mathbf{U}$. If there exists a partition of $\left\lbrace 1,2,...,L \right\rbrace$ where $L \leq 2\kappa$ into $g$ non empty subsets $\Gamma_1, \Gamma_2,...,\Gamma_g$ with cardinalities $\kappa_1,...,\kappa_g$ such that $U_{ij}=0$ whenever $i \in \Gamma_p$ and $j \in \Gamma_q$ and $p \neq q$, and if any group $\Gamma_i$ admits fast decodability, then the code is fast group decodable.
\end{lemma}

These lemmas state that the HRQF matrix totally determines the fast sphere decodability of STBCs and provide sufficient conditions for an STBC to be multi-group, fast or fast-group decodable. Nevertheless, as highlighted in \cite{Jithamithra10,Jithamithra11,Jithamithra13}, in some cases, for instance in the case of Block orthogonal codes, the HRQF approach does not capture the zero structure of the $\mathbf{R}$ matrix. In such cases, it is possible to have entries $R_{ij} \neq 0$ even if the corresponding weight matrices $\mathbf{A}_i$ and $\mathbf{A}_j$ are HR orthogonal which is equivalent to have the corresponding entry in the HRQF matrix $U_{ij}= 0$. Authors in \cite{Jithamithra10,Jithamithra11,Jithamithra13} do not provide an explanation for having such configurations. 

\section{Sufficient Design Criteria for FSD of STBCs}
In this work, we first aim to provide sufficient conditions on the structural properties of the weight matrices for an STBC that fully determine the FSD of any STBC. These design criteria are stated in theorem \ref{ourth}.

\begin{theorem}\label{ourth}
For an STBC with $k$ independent complex symbols and $2k$ weight matrices $\mathbf{A}_l$ for $l=1,...,2\kappa$, if for any $i$ and $j$, $i\neq j, 1 \leq i,j \leq 2\kappa$, one or both of the following conditions is satisfied:
\vspace{-0.5cm}
\begin{equation}\label{c1}
\sum_{l=1}^{T}\mathrm{Tr}\left( a_{ql}^{(i)}\left(a_{pl}^{(j)} \right)^{*} + a_{pl}^{(i)}\left(a_{ql}^{(j)} \right)^{*} \right) = 0 ~,~ \forall q=1,...,n_t, p=q,q+1,...,n_t  
\end{equation}
\begin{equation}\label{c2}
\sum_{l=1}^{T}\mathrm{Tr}\left(i\left[ a_{ql}^{(i)}\left(a_{pl}^{(j)} \right)^{*} - a_{pl}^{(i)}\left(a_{ql}^{(j)} \right)^{*} \right]\right) = 0 ~,~ \forall q=1,...,n_t, p=q+1,...,n_t  
\end{equation}
where $a_{ql}^{(i)}$ (resp. $a_{pl}^{(j)}$) is the entry of the matrix $\mathbf{A}_i$ (resp. $\mathbf{A}_j$) at row $q$ and column $l$ (resp. at row $p$ and column $l$), then the $i^{th}$ and $j^{th}$ columns of the equivalent channel matrix $\mathbf{H}_{eq}$ are orthogonal. Both conditions hold at the same time if and only if $a_{ql}^{(i)}$ or $a_{pl}^{(j)}=0$.
\end{theorem}
\begin{proof}
We know from the expression of the equivalent channel matrix that:
\begin{equation}
<\mathbf{h}_{i}^{eq}, \mathbf{h}_{j}^{eq}> = \sum_{l=1}^{T} \tilde{\mathbf{a}}_{il}^{t} \check{\mathbf{H}}_{eq}^{t} \check{\mathbf{H}}_{eq} \tilde{\mathbf{a}}_{jl} = \sum_{l=1}^{T} \tilde{\mathbf{a}}_{il}^{t} \mathbf{M} \tilde{\mathbf{a}}_{jl}
\end{equation}
where $\mathbf{a}_{il}, \mathbf{a}_{jl}, l=1,...,T$ are the $l^{th}$ columns of respectively the weight matrix $\mathbf{A}_i$ and $\mathbf{A}_j$. 
%Using the operator $(\tilde{.})$, we prove in appendix (\ref{append1}) that the matrix $\mathbf{M}$ is symmetric and satistifies the following properties:
 Let $T_l=\tilde{\mathbf{a}}_{il}^{t} \mathbf{M} \tilde{\mathbf{a}}_{jl}$. By deriving the computation and using the properties that $\Re{(x)}\Re{(y)}+\Im{(x)}\Im{(y)}=\Re{(xy^{*})}$ and $\Re{(x)}\Im{(y)}-\Im{(x)}\Re{(y)}=\Re{(ixy^{*})}$ we get:
 \begin{align}
 T_{l} = \underbrace{\frac{1}{2} \sum_{q=1}^{n_t}\sum_{p=1}^{n_t}M_{2p-1,2q-1} \mathrm{Tr}\left( a_{ql}^{(i)}\left(a_{pl}^{(j)} \right)^{*} \right)}_{A_l} + \underbrace{\frac{1}{2}\sum_{q=1}^{n_t} \sum_{p=1}^{n_t} M_{2q-1,2p} \mathrm{Tr}\left( i a_{ql}^{(i)}\left(a_{pl}^{(j)} \right)^{*} \right)}_{B_l} \notag
 \end{align}
 Notice that the terms $\mathrm{Tr}\left( a_{ql}^{(i)}\left(a_{pl}^{(j)} \right)^{*} \right)$ in $A_l$ and $\mathrm{Tr}\left( i a_{ql}^{(i)}\left(a_{pl}^{(j)} \right)^{*} \right)$ in $B_l$ are equal to zeros at the same time if and only if either $ a_{ql}^{(i)}=0$ or $ a_{pl}^{(j)}$. Now in order to simplify the terms $A_l$ and $B_l$ we use the properties of the matrix $\mathbf{M}$ proved in appendix \ref{append1}. We obtain for $A_l$, using the symmetry of the matrix $\mathbf{M}$, the following: 
 \begin{align}
 A_l &= \frac{1}{2} \sum_{q=1}^{n_t}M_{2q-1,2q-1} \mathrm{Tr}\left( a_{ql}^{(i)}\left(a_{ql}^{(j)} \right)^{*} \right) + \frac{1}{2}\sum_{q=1}^{n_t} \sum_{\tiny\begin{array}{c}
 p=1 \\ p \neq q\end{array}}^{n_t} M_{2q-1,2p-1} \mathrm{Tr}\left( a_{ql}^{(i)}\left(a_{pl}^{(j)} \right)^{*} \right) \notag \\
 &= \frac{1}{2} \sum_{q=1}^{n_t}M_{2q-1,2q-1} \mathrm{Tr}\left( a_{ql}^{(i)}\left(a_{ql}^{(j)} \right)^{*} \right) + \frac{1}{2}\sum_{q=1}^{n_t} \sum_{p=q+1}^{n_t}M_{2q-1,2p-1}\mathrm{Tr}\left( a_{ql}^{(i)}\left(a_{pl}^{(j)} \right)^{*} + a_{pl}^{(i)}\left(a_{ql}^{(j)} \right)^{*} \right) \notag
 \end{align}
And for the term $B_l$, we use the symmetry of $\mathbf{M}$ and the properties that $M_{2q-1,2q}=0, \forall q=1,...,n_t$ and $M_{2q-1,2p}=-M_{2q,2p-1}, \forall 1 \leq q < p \leq n_t$ and obtain:
 \begin{align}
 B_l &= \frac{1}{2}\sum_{q=1}^{n_t} \sum_{\tiny\begin{array}{c}
 p=1 \\ p \neq q\end{array}}^{n_t} M_{2q-1,2p} \mathrm{Tr}\left( i a_{ql}^{(i)}\left(a_{pl}^{(j)} \right)^{*} \right) = \frac{1}{2}\sum_{q=1}^{n_t}\sum_{p=q+1}^{n_t}M_{2q-1,2p}\mathrm{Tr}\left(i\left[ a_{ql}^{(i)}\left(a_{pl}^{(j)} \right)^{*} - a_{pl}^{(i)}\left(a_{ql}^{(j)} \right)^{*} \right]\right) \notag
 \end{align}
The proof of the theorem follows given $<\mathbf{h}_{i}^{eq}, \mathbf{h}_{j}^{eq}> = \displaystyle\sum_{l=1}^{T} T_l$.
\end{proof}

Theorem \ref{ourth} states a component-wise mutual orthogonality criterion involving the entries of weight matrices corresponding to column vectors of the equivalent channel matrix. In order to show the suboptimality of the HRQF-based approach, we first provide in the following lemma the design conditions we derive based on the mutual orthogonality property of weight matrices as proposed in Theorem (\ref{HRT}) and compare them subsequently to the conditions considered in the HRQF-based method explaining why the latter does not allow to capture all the families of low-complexity ML decoding STBCs in contrast to our derived sufficient design conditions in Theorem \ref{ourth}.

\begin{lemma}\label{var1}
Consider an STBC with $\kappa$ independent complex information symbols and $2\kappa$ weight matrices $\mathbf{A}_l, l=1,...,2\kappa$. For any $i$ and $j$, $1 \leq i,j \leq 2k$ the matrices $\mathbf{A}_i$ and $\mathbf{A}_j$ are mutually orthogonal, i.e. satisfy $\mathbf{A}_i \mathbf{A}_j^{H}+\mathbf{A}_j\mathbf{A}_i^{H}=\mathbf{0}_{n_t}$ if and only if both following conditions are met:
\vspace{-0.45cm}
\begin{align}\label{cond}
 \sum_{l=1}^{T}\mathrm{Tr}\left( a_{pl}^{(i)}\left(a_{ql}^{(j)} \right)^{*} \right)=0 ~,~\forall p=1,...,n_t ~;~ q=p,...,n_t \\
 \sum_{l=1}^{T}\mathrm{Tr}\left( i a_{pl}^{(i)}\left(a_{ql}^{(j)} \right)^{*} \right)=0 ~,~\forall p=1,...,n_t, q=p,...,n_t 
\end{align}
\end{lemma}
\begin{proof}
We know from the properties of the $(\check{.})$ operation that $\mathbf{A}=\mathbf{B}\mathbf{C} \Leftrightarrow \check{\mathbf{A}}=\check{\mathbf{B}}\check{\mathbf{C}}$. Using this property, we have:
\begin{align}\label{property}
\mathbf{A}_i \mathbf{A}_j^{H}+\mathbf{A}_j\mathbf{A}_i^{H}=\mathbf{0}_{n_t} & \Leftrightarrow \check{\mathbf{A}}_i\left(\check{\mathbf{A}}_j \right)^{t} + \check{\mathbf{A}}_j\left(\check{\mathbf{A}}_i \right)^{t} = \mathbf{O}_{2n_t} 
\end{align}
Let the matrix $\mathbf{V}=\check{\mathbf{A}}_i\left(\check{\mathbf{A}}_j \right)^{t}+\check{\mathbf{A}}_j\left(\check{\mathbf{A}}_i \right)^{t}$. Then, equation (\ref{property}) is equivalent to $\mathbf{V}=\mathbf{0}_{2n_t}$. $\mathbf{V}$ is symmetric and its entries are given by:
\begin{align}
& V_{2p-1,2p-1} = V_{2p,2p} = \sum_{l=1}^{T} \mathrm{Tr}\left(a_{pl}^{(i)}\left(a_{pl}^{(j)} \right)^{*}\right) ~,~ \forall p=1...,n_t \notag \\
& V_{2p-1,2q} = V_{2p,2q-1}= \sum_{l=1}^{T}\mathrm{Tr}\left(i a_{pl}^{(i)}\left(a_{ql}^{(j)} \right)^{*}\right) ~,~ \forall p=1,...,n_t ~;~ q=p+1,...,n_t\notag \\
& V_{2p-1,2q-1} = V_{2p,2q}= \sum_{l=1}^{T}\mathrm{Tr}\left(a_{pl}^{(i)}\left(a_{ql}^{(j)} \right)^{*}\right) ~,~ p=1,...,n_t ~;~ q= p+1,...,n_t\notag 
\end{align}
Having the matrix $\mathbf{V}$ equal to 0 ends the proof.
\end{proof}
Now for what concerns the design criteria proposed in literature based on HR-theory and the HRQF matrix, authors in \cite{Jithamithra10,Jithamithra13} and mainly in Theorem 2 in \cite{Srinath09} provide a sufficient condition and use in their proof that if $\mathbf{A}_i\mathbf{A}_j^{H}+\mathbf{A}_j\mathbf{A}_i^{H}=\mathbf{0}_{n_t}$ then the matrix $\check{\mathbf{A}}_i\left(\check{\mathbf{A}}_j \right)^{t}$ is skew-symmetric, thus its diagonal elements are zeros. Accordingly, the conditions to obtain orthogonality between the corresponding columns of the equivalent channel matrix are given by: $V_{2p-1,2p-1} = V_{2p,2p}= \sum_{l=1}^{T}  \mathrm{Tr}\left( a_{pl}^{(i)}\left(a_{pl}^{(j)} \right)^{*} \right) =0$.

Thus, compared to the conditions proposed in Lemma \ref{var1}, we can easily deduce the suboptimality of these conditions: first the sufficient conditions consider only the diagonal elements, second, even the criteria derived in Lemma \ref{var1} impose conditions more than required to obtain orthogonality between two columns in the equivalent channel matrix as proposed using the HR mutual orthogonality and HRQF approach. These criteria capture the summation of the trace forms of the components $a_{pl}^{(i)}\left(a_{pl}^{(j)} \right)^{*}$ and impose that this summation be zero. However, as proved in Theorem \ref{ourth}, in order to have the $i^{th}$ and $j^{th}$ columns of $\mathbf{H}_{eq}$ orthogonal, it is sufficient to have the crossed trace forms for $\left( a_{pl}^{(i)}\left(a_{pl}^{(j)} \right)^{*}\pm a_{ql}^{(i)}\left(a_{pl}^{(j)} \right)^{*} \right)$ equal to 0. In such cases, the HR mutual orthogonality is satisfied and the entry $U_{ij}$ of the HRQF matrix is equal to 0 without having orthogonality of columns $i$ and $j$ of the equivalent channel matrix and thus the corresponding entry in the $\mathbf{R}$ matrix $R_{ij} \neq 0$.

%\vspace{0.25cm}
Having derived the sufficient design criteria for having orthogonality of columns in the equivalent channel matrix, we apply these conditions and show that these criteria are enough to determine the FSD complexity of an STBC, prove that the FSD complexity depends only on the weight matrices and their ordering and not on the channel matrix or the number of receive antennas. We start our study with the class of multi-group decodable codes.

\begin{lemma}\label{groupdlema}
Let an STBC with $\kappa$ independent complex symbols and $2\kappa$ weight matrices. If there exists an ordered partition of $\left\lbrace 1,2,...,2\kappa \right\rbrace$ into $g$ non empty subsets $\Gamma_1,...,\Gamma_g$ such that at least one of the following conditions is satisfied:
\begin{equation}\label{conditions11gd}
\sum_{l=1}^{T}\mathrm{Tr}\left( a_{ql}^{(i)}\left(a_{pl}^{(j)} \right)^{*} + a_{pl}^{(i)}\left(a_{ql}^{(j)} \right)^{*} \right) = 0 ~,~ \forall q=1,...,n_t, p=q,q+1,...,n_t  
\end{equation}
\begin{equation}\label{conditions22gd}
\sum_{l=1}^{T}\mathrm{Tr}\left(i\left[ a_{ql}^{(i)}\left(a_{pl}^{(j)} \right)^{*} - a_{pl}^{(i)}\left(a_{ql}^{(j)} \right)^{*} \right]\right) = 0 ~,~ \forall q=1,...,n_t, p=q+1,...,n_t  
\end{equation}
 whenever $i \in \Gamma_m$ and $j \in \Gamma_n$ and $m \neq n$, then the code is $g-$group sphere decodable.
\end{lemma}

\begin{proof}
Let $\mathbf{R}$ be the matrix obtained from the QR decomposition of the equivalent channel matrix $\mathbf{H}_{eq}=\mathbf{Q}\mathbf{R}$. In order to show that the code is $g-$group decodable, we need to show that $R_{ij}=0$ whenever $i \in \Gamma_m$ and $j \in \Gamma_n$ and $m \neq n$. We know from theorem \ref{ourth} that if equations (\ref{conditions11gd}) and (\ref{conditions22gd})  are satisfied, then the $i^{th}$ and $j^{th}$ columns of $\mathbf{H}_{eq}$ are orthogonal, i.e. $<\mathbf{h}_i^{eq},\mathbf{h}_{j}^{eq}>=0$. Let $L_{u}=\sum_{v=1}^{u}|\Gamma_v|$ where $u=1,2,...,g$ and $L_0=0$. For any group $\Gamma_u$, we need to show that $R_{ij}=0$ for $L_{u-1}+1 \leq i \leq L_u$ and $L_u+1 \leq j \leq 2\kappa$. The proof is performed by induction. Consider the first group $\Gamma_1$. We have, for $1 \leq i \leq L_1$ and any $L_1 +1 \leq j \leq 2\kappa$:
\begin{equation}\label{conditions111}
\sum_{l=1}^{T}\mathrm{Tr}\left( a_{ql}^{(i)}\left(a_{pl}^{(j)} \right)^{*} + a_{pl}^{(i)}\left(a_{ql}^{(j)} \right)^{*} \right) = 0 ~,~ \forall q=1,...,n_t, p=q,q+1,...,n_t  
\end{equation}
\begin{equation}\label{conditions222}
\sum_{l=1}^{T}\mathrm{Tr}\left(i\left[ a_{ql}^{(i)}\left(a_{pl}^{(j)} \right)^{*} - a_{pl}^{(i)}\left(a_{ql}^{(j)} \right)^{*} \right]\right) = 0 ~,~ \forall q=1,...,n_t, p=q+1,...,n_t  
\end{equation}
We need to show that $R_{ij}=0$. For $i=1$ and any $j \geq L_1+1$, we have:
\vspace{-0.25cm}
\begin{equation}
R_{1j} = <\mathbf{q}_1,\mathbf{h}^{eq}_j> = \frac{1}{\parallel \mathbf{h}_1^{eq}\parallel}<\mathbf{h}^{eq}_1,\mathbf{h}^{eq}_j> = 0
\end{equation} 
given that $\mathbf{q}_1=\frac{1}{\parallel \mathbf{h}_1^{eq}\parallel}\mathbf{h}_1^{eq}$ and equations (\ref{conditions111}) and (\ref{conditions222}) are satisfied. Now let the induction hypothesis be true, i.e. let $<\mathbf{q}_k,\mathbf{h}^{eq}_j>=0$ for all $k < i$ for any $i$ such that $1 \leq i \leq L_1$. We have:
%\vspace{-0.5cm}
\begin{align}
<\mathbf{q}_i,\mathbf{h}^{eq}_j> &= \frac{1}{\parallel \mathbf{r}_i\parallel} \left[ <\mathbf{h}^{eq}_i,\mathbf{h}^{eq}_j> - \displaystyle\sum_{k=1}^{i-1}<\mathbf{q}_k,\mathbf{h}^{eq}_i><\mathbf{q}_k,\mathbf{h}^{eq}_j> \right]= 0 \notag 
\end{align}
Given that $<\mathbf{h}^{eq}_i,\mathbf{h}^{eq}_j>$ from the conditions in equations (\ref{conditions11gd}) and (\ref{conditions22gd}) and $<\mathbf{q}_k,\mathbf{h}^{eq}_i>=0$ for $k < i$ by induction hypothesis.
Now we consider the $t^{th}$ group $\Gamma_t$. Let the induction hypothesis hold for all groups $1,2,...,t-1$. Consider $R_{ij}$ with $L_{t-1}+1 \leq i \leq  L_t$ and $L_t+1 \leq j \leq 2\kappa$. We have,
\vspace{-0.5cm}
\begin{align}
R_{ij} &= <\mathbf{q}_i,\mathbf{h}^{eq}_j> = \frac{1}{\parallel \mathbf{r}_i\parallel} \left[ <\mathbf{h}^{eq}_i - \displaystyle\sum_{k=1}^{i-1}<\mathbf{q}_k,\mathbf{h}^{eq}_i>\mathbf{q}_k,\mathbf{h}^{eq}_j> \right] \notag \\
&= \frac{1}{\parallel \mathbf{r}_i\parallel} \left[ <\mathbf{h}^{eq}_i,\mathbf{h}^{eq}_j> - \displaystyle\sum_{k=1}^{i-1}<\mathbf{q}_k,\mathbf{h}^{eq}_i><\mathbf{q}_k,\mathbf{h}^{eq}_j> \right] = 0
\end{align}
given that $<\mathbf{h}^{eq}_i,\mathbf{h}^{eq}_j>=0$ from conditions in equations (\ref{conditions11gd}) and (\ref{conditions22gd}) and $<\mathbf{q}_k,\mathbf{h}^{eq}_i>=0$ for $k < i$ by the induction hypothesis.
\end{proof}
\vspace{-0.75cm}
\paragraph*{\textbf{Example 5:} we aim here, through the example of the ABBA code, to validate our lemma \ref{groupdlema} for multi-group codes construction. For this purpose and given the codeword matrix in equation (\ref{mat_abba}), we need first to explicitly write the weight matrices of the ABBA code. As the symbols $x_i, i=1,...,4$ are reals, they correspond to real and imaginary parts of a 2 encoded complex-valued symbols $s_1$ and $s_2$ as: $x_1 = \Re{(s_1)} ~,~ x_2 = \Im{(s_1)} ~,~ x_3 = \Re{(s_2)} ~,~ x_4 = \Im{(s_2)}$. Accordingly, the codeword matrix can be written as:
\begin{align}
\mathbf{X}_{ABBA} &= \left[ \begin{array}{cc}
\Re{(s_1)}+i\Im{(s_2)} & -\Im{(s_1)}+i\Re{(s_2)} \\
-\Im{(s_1)}+i\Re{(s_2)} & \Re{(s_1)} + i \Im{(s_2)} 
\end{array} \right] = \Re{(s_1)}\mathbf{A}_1 + \Im{(s_1)}\mathbf{A}_2 + \Re{(s_2)}\mathbf{A}_3 + i \Im{(s_2)}\mathbf{A}_4 \notag
\end{align}
\vspace{-0.25cm}
Then the $4$ weight matrices are given by:
\begin{align}
\mathbf{A}_1 = \left[ \begin{array}{cc}
1 & 0 \\ 0 & 1
\end{array} \right] ~,~ \mathbf{A}_2 = \left[ \begin{array}{cc}
0 & -1 \\ -1 & 0
\end{array} \right] ~,~ \mathbf{A}_3 = \left[ \begin{array}{cc}
0 & i \\ i & 0
\end{array} \right] ~,~ \mathbf{A}_4 = \left[ \begin{array}{cc}
i & 0 \\ 0 & i
\end{array} \right]
\end{align}
Now, in order to determine the zero structure of the corresponding $\mathbf{R}$ matrix, we check the design conditions according to the lemma (\ref{groupdlema}). We have the following:
\begin{align}
\left\lbrace \mathbf{A}_1, \mathbf{A}_3 \right\rbrace : & l=1 : \mathrm{Tr}\left(1(0)^{*}\right)=0,\mathrm{Tr}\left(0(i)^{*}\right)=0    ,\mathrm{Tr}\left(1(i)^{*}+0(0)^{*}\right)=0\notag \\
& l=2 : \mathrm{Tr}\left(0(i)^{*}\right)=0,\mathrm{Tr}\left(1(0)^{*}\right)=0  ,\mathrm{Tr}\left(0(0)^{*}+1(i)^{*}\right)=0 \notag 
\end{align}
\begin{align}
\left\lbrace \mathbf{A}_1, \mathbf{A}_4 \right\rbrace : & l=1 : \mathrm{Tr}\left(1(i)^{*}\right)=0,\mathrm{Tr}\left(0(0)^{*}\right)=0,\mathrm{Tr}\left(1(0)^{*}+0(i)^{*}\right)=0 \notag \\
& l=2 : \mathrm{Tr}\left(0(0)^{*}\right)=0,\mathrm{Tr}\left(1(i)^{*}\right)=0,\mathrm{Tr}\left(0(i)^{*}+1(0)^{*}\right)=0 \notag 
\end{align}
\begin{align}
\left\lbrace \mathbf{A}_2, \mathbf{A}_3 \right\rbrace : & l=1 : \mathrm{Tr}\left(0(0)^{*}\right)=0,\mathrm{Tr}\left((-1)(i)^{*}\right)=0,\mathrm{Tr}\left(0(i)^{*}+(-1)(0)^{*}\right)=0 \notag \\
& l=2 : \mathrm{Tr}\left((-1)(i)^{*}\right)=0,\mathrm{Tr}\left(0(0)^{*}\right)=0,\mathrm{Tr}\left((-1)(0)^{*}+0(i)^{*}\right)=0 \notag 
\end{align}
\begin{align}
\left\lbrace \mathbf{A}_2, \mathbf{A}_4 \right\rbrace : & l=1 : \mathrm{Tr}\left(0(i)^{*}\right)=0,\mathrm{Tr}\left((-1)(0)^{*}\right)=0,\mathrm{Tr}\left(0(0)^{*}+(-1)(i)^{*}\right)=0 \notag \\
& l=2 : \mathrm{Tr}\left((-1)(0)^{*}\right)=0,\mathrm{Tr}\left(0(i)^{*}\right)=0,\mathrm{Tr}\left((-1)(i)^{*}+0(0)^{*}\right)=0 \notag 
\end{align}
Notice that for each couple of matrices, the individual trace forms are equal to zero for each of the columns $l=1$ and $l=2$, thus the summation of the trace forms over the two columns is equal to zero which validates the conditions of (\ref{groupdlema}). Thus, there exists an ordered partition of $\left\lbrace 1,...,4 \right\rbrace$ into $2$ non-empty subsets $\Gamma_1=\left\lbrace \mathbf{A}_1,\mathbf{A}_2 \right\rbrace$ and $\Gamma_2=\left\lbrace \mathbf{A}_3,\mathbf{A}_4 \right\rbrace$ such that for all $l=1,2$ the conditions of equation (\ref{conditions11gd}) are satisfied, then the code is $2-$group decodable. This zero structure confirms the form of the matrix $\mathbf{R}$ plotted through figure (\ref{abbafig}) obtained via simulations.  
}
\begin{lemma}\label{fastd}
Let an STBC with $\kappa$ independent complex symbols and $2\kappa$ weight matrices. If there exists a partition of $\left\lbrace 1,2,...,L \right\rbrace$ where $L \leq 2\kappa$ into $g$ non empty subsets $\Gamma_1, \Gamma_2,...,\Gamma_g$ such that:
\vspace{-0.25cm}
\begin{equation}\label{conditions2}
\sum_{l=1}^{T}\mathrm{Tr}\left( a_{ql}^{(i)}\left(a_{pl}^{(j)} \right)^{*} + a_{pl}^{(i)}\left(a_{ql}^{(j)} \right)^{*} \right) = 0 ~,~ \forall q=1,...,n_t, p=q,q+1,...,n_t  
\end{equation}
\begin{equation}\label{conditions2bis}
\sum_{l=1}^{T}\mathrm{Tr}\left(i\left[ a_{ql}^{(i)}\left(a_{pl}^{(j)} \right)^{*} - a_{pl}^{(i)}\left(a_{ql}^{(j)} \right)^{*} \right]\right) = 0 ~,~ \forall q=1,...,n_t, p=q+1,...,n_t  
\end{equation}
 whenever $i \in \Gamma_m$ and $j \in \Gamma_n$ and $m \neq n$, then the code is fast decodable.
\end{lemma}

\begin{proof}
The proof follows from the proof of Lemma (\ref{groupdlema}) by replacing $2\kappa$ with $L$ in the ordered partition.
\end{proof}

\paragraph*{\textbf{Example 6}: we aim in this example to illustrate the above lemma through the example of the Silver code described in example 2. The $8$ weight matrices are given by: 
\vspace{-0.25cm}
\small{
 \begin{align}
 & \mathbf{A}_1= \left[ \begin{array}{cc}
 1 & 0 \\ 0 & 1
  \end{array}\right] ~, \mathbf{A}_2= \left[ \begin{array}{cc}
  i & 0 \\ 0 & -i
   \end{array}\right] ~,~ \mathbf{A}_3= \left[ \begin{array}{cc}
   0 & -1 \\ 1 & 0
    \end{array}\right] ~,~ \mathbf{A}_4= \left[ \begin{array}{cc}
   0 & i \\ i & 0
     \end{array}\right] \notag \\
      & \mathbf{A}_5= \left[ \begin{array}{cc}
        U_{11} & -U_{21}^{*} \\ -U_{21} & -U_{11}^{*}
         \end{array}\right] ~, \mathbf{A}_6= \left[ \begin{array}{cc}
         iU_{11} & iU_{21}^{*} \\ -iU_{21} & iU_{11}^{*}
          \end{array}\right] ~,~ \mathbf{A}_7= \left[ \begin{array}{cc}
          U_{12} & -U_{22}^{*} \\ -U_{22} & -U_{12}^{*}
           \end{array}\right] ~,~ \mathbf{A}_8= \left[ \begin{array}{cc}
          iU_{12} & iU_{22}^{*} \\ -iU_{22} & iU_{12}^{*}
            \end{array}\right] \notag
    \end{align} }
    \normalsize
    %\vspace{-0.25cm}
    with $U_{11}=\frac{1}{\sqrt{7}}\left(1+i\right), U_{12}=\frac{1}{\sqrt{7}}\left(-1+2i\right), U_{21}=\frac{1}{\sqrt{7}}\left(1+2i\right)=-U_{12}^{*}, U_{22}=\frac{1}{\sqrt{7}}\left(1-i\right)=U_{11}^{*}$. Now, in order to determine the zero structure of the corresponding $\mathbf{R}$ matrix, we check the design conditions according to lemma (\ref{fastd}). We distinguish two parts. For the first part, the conditions in equation (\ref{conditions2}) are met as follows: 
    \vspace{-0.25cm}
   \begin{align}
   \left\lbrace \mathbf{A}_1,\mathbf{A}_2 \right\rbrace ~ & l=1: \mathrm{Tr}\left(1(i)^{*}\right)=0,\mathrm{Tr}\left(0(0)^{*}\right)=0,\mathrm{Tr}\left(1(0)^{*}+0(i)^{*}\right)=0 \notag \\
   & l=2 : \mathrm{Tr}\left(0(0)^{*}\right)=0,\mathrm{Tr}\left(1(-i)^{*}\right)=0, \mathrm{Tr}\left(0(-i)^{*}+1(0)^{*}\right)=0  \notag
   \end{align}
    \begin{align}
    \left\lbrace \mathbf{A}_1,\mathbf{A}_4 \right\rbrace ~ & l=1: \mathrm{Tr}\left(1(0)^{*}\right)=0,\mathrm{Tr}\left(0(i)^{*}\right)=0, \mathrm{Tr}\left(1(i)^{*}+0(0)^{*}\right)=0 \notag \\
    & l=2 : \mathrm{Tr}\left(0(i)^{*}\right)=0,\mathrm{Tr}\left(1(0)^{*}\right)=0,\mathrm{Tr}\left(0(0)^{*}+1(i)^{*}\right)=0  \notag
    \end{align}
    \begin{align}
     \left\lbrace \mathbf{A}_2,\mathbf{A}_3 \right\rbrace ~ & l=1: \mathrm{Tr}\left(i(0)^{*}\right)=0,\mathrm{Tr}\left(0(1)^{*}\right)=0, \mathrm{Tr}\left(i(1)^{*}+0(0)^{*}\right)=0 \notag \\
     & l=2 : \mathrm{Tr}\left(0(-1)^{*}\right)=0,\mathrm{Tr}\left((-i)(0)^{*}\right)=0,\mathrm{Tr}\left(0(0)^{*}+(-i)(-1)^{*}\right)=0  \notag
     \end{align}
    \begin{align}
     \left\lbrace \mathbf{A}_3,\mathbf{A}_4 \right\rbrace ~ & l=1: \mathrm{Tr}\left(0(0)^{*}\right)=0,\mathrm{Tr}\left(1(i)^{*}\right)=0, \mathrm{Tr}\left(0(i)^{*}+1(0)^{*}\right)=0 \notag \\
     & l=2 : \mathrm{Tr}\left((-1)(i)^{*}\right)=0,\mathrm{Tr}\left(0(0)^{*}\right)=0,\mathrm{Tr}\left((-1)(0)^{*}+0(i)^{*}\right)=0  \notag
     \end{align}  
     \begin{align}
        &  \left\lbrace \mathbf{A}_5,\mathbf{A}_6 \right\rbrace ~ l=1: \mathrm{Tr}\left(U_{11}(iU_{11})^{*}\right)=0,\mathrm{Tr}\left((-U_{21})(-iU_{21})^{*}\right)=0, \mathrm{Tr}\left(U_{11}(-iU_{21})^{*}+(-U_{21})(iU_{11})^{*}\right)=0 \notag \\
          & l=2 : \mathrm{Tr}\left((-U_{21})^{*}(-iU_{21})^{*}\right)=0,\mathrm{Tr}\left((-U_{11})^{*}(-iU_{11})^{*}\right)=0,\mathrm{Tr}\left((-U_{21})^{*}(-iU_{11})^{*}+(-U_{11}^{*})(-iU_{21})^{*}\right)=0  \notag
          \end{align} 
     \begin{align}
          &  \left\lbrace \mathbf{A}_7,\mathbf{A}_8 \right\rbrace ~ l=1: \mathrm{Tr}\left(U_{12}(iU_{12})^{*}\right)=0,\mathrm{Tr}\left((-U_{22})(-iU_{22})^{*}\right)=0, \mathrm{Tr}\left(U_{12}(-iU_{22})^{*}+(-U_{22})(iU_{12})^{*}\right)=0 \notag \\
            & l=2 : \mathrm{Tr}\left((-U_{22})^{*}(-iU_{22})^{*}\right)=0,\mathrm{Tr}\left((-U_{12})^{*}(-iU_{12})\right)=0,\mathrm{Tr}\left((-U_{22})^{*}(-iU_{12})+(-U_{12}^{*})(-iU_{22})\right)=0  \notag
            \end{align}         
  Then, for the second part, conditions in equation (\ref{conditions2bis}) are satisfied according to:
       \begin{align}
       \left\lbrace \mathbf{A}_1,\mathbf{A}_3 \right\rbrace ~ & l=1: \mathrm{Tr}\left(i\left[ 1(1)^{*}-0(0)^{*} \right] \right)=0 ~;~ l=2 : \mathrm{Tr}\left(i\left[ 0(0)^{*}-1(-1)^{*} \right] \right)=0 \notag
       \end{align}
         \begin{align}
         \left\lbrace \mathbf{A}_2,\mathbf{A}_4 \right\rbrace ~& l=1: \mathrm{Tr}\left(i\left[ i(i)^{*}-0(0)^{*} \right] \right)=0 ~;~ l=2: \mathrm{Tr}\left(i\left[ 0(0)^{*}-(-i)(i)^{*} \right] \right)=0 \notag
         \end{align}    
         Notice that for each couple of the above analyzed matrices, the individual trace forms are equal to zero for each of the columns $l=1$ and $l=2$, thus the summation of the trace forms over the two columns is equal to zero. For the remaining matrices we have the following:
         \footnotesize{
 \begin{align}
  & \left\lbrace \mathbf{A}_5,\mathbf{A}_7 \right\rbrace ~ l=1,2 :~ \notag \\  &\mathrm{Tr}\left(i \left[-U_{11}U_{22}^{*}+U_{21}U_{12}^{*} \right] + i \left[ U_{21}^{*}U_{12}-U_{11}^{*}U_{22} \right]\right) = \mathrm{Tr}\left(-i\left[U_{11}U_{22}^{*}+U_{11}^{*}U_{22} \right] + i\left[U_{21}U_{12}^{*}+U_{21}^{*}U_{12} \right] \right) =0 \notag
 \end{align}  
  \begin{align}
  &\left\lbrace \mathbf{A}_5,\mathbf{A}_8 \right\rbrace ~ l=1,2 :~ \notag \\ & \mathrm{Tr}\left(i \left[U_{11}(iU_{22}^{*})+U_{21}(-iU_{12}^{*}) \right] + i \left[ (-U_{21}^{*})(-iU_{12})-(-U_{11}^{*})(-iU_{22}) \right]\right) = \mathrm{Tr}\left(-\left[U_{11}U_{22}^{*}-U_{11}^{*}U_{22} \right] - \left[U_{21}^{*}U_{12}-U_{21}U_{12}^{*} \right] \right) =0 \notag
  \end{align}
   \begin{align}
    &\left\lbrace \mathbf{A}_6,\mathbf{A}_7 \right\rbrace l=1,2 :~ \notag \\ & \mathrm{Tr}\left(i \left[iU_{11}(-U_{22}^{*})-(-iU_{21})(U_{12}^{*}) \right] + i \left[ (iU_{21}^{*})(-U_{12})-(iU_{11}^{*})(-U_{22}) \right]\right) = \mathrm{Tr}\left(-\left[U_{11}U_{22}^{*}-U_{11}^{*}U_{22} \right] - \left[U_{21}U_{12}^{*}-U_{21}^{*}U_{12} \right] \right) =0 \notag
    \end{align}
      \begin{align}
        &\left\lbrace \mathbf{A}_6,\mathbf{A}_8 \right\rbrace l=1,2 :~ \notag \\ & \mathrm{Tr}\left(i \left[iU_{11}(iU_{22}^{*})+(iU_{21})(-iU_{12}^{*}) \right] + i \left[ (iU_{21}^{*})(-iU_{12})-(iU_{11}^{*})(-iU_{22}) \right]\right) = \mathrm{Tr}\left(-i\left[U_{11}U_{22}^{*}+U_{11}^{*}U_{22} \right] + i\left[U_{21}U_{12}^{*}+U_{21}^{*}U_{12} \right] \right) =0 \notag
        \end{align}}
        \normalsize
The second group of weight matrices satisfy the conditions of equation (\ref{conditions2bis}). We deduce the zero structure of the $\mathbf{R}$ matrix corresponding to the Silver code as depicted in figure \ref{sc}.
 }
\begin{remark}
The Silver code is an example of fast decodable codes and particularly Block orthogonal code. In order to illustrate the suboptimality of the HRQF-based method, we check the conditions in Lemma (\ref{var1}) for the matrices $\lbrace \mathbf{A}_5,\mathbf{A}_7 \rbrace$. We notice that the condition of equation (\ref{cond}) for $p=1, q=2$ is not satisfied since $\mathrm{Tr}\left(U_{11}(-U_{22})^{*}+(U_{21}^{*}U_{12})\right) \neq 0$. Using the HRQF-based method in this case, the corresponding entry in the HRQF matrix is equal to zero while the corresponding value in the matrix $\mathbf{R}$ is different from zero. The zero structure of the $\mathbf{R}$ matrix is not then fully determined by the HRQF-based design criteria. 
\end{remark}
\begin{lemma}
Let an STBC with $\kappa$ independent complex symbols, $2\kappa$ weight matrices and HRQF matrix $\mathbf{U}$. If there exists a partition of $\left\lbrace 1,2,...,L \right\rbrace$ where $L \leq 2\kappa$ into $g$ non empty subsets $\Gamma_1, \Gamma_2,...,\Gamma_g$ with cardinalities $\kappa_1,...,\kappa_g$ such that 
\begin{equation}\label{conditions3}
\sum_{l=1}^{T}\mathrm{Tr}\left( a_{ql}^{(i)}\left(a_{pl}^{(j)} \right)^{*} + a_{pl}^{(i)}\left(a_{ql}^{(j)} \right)^{*} \right) = 0 ~,~ \forall q=1,...,n_t, p=q,q+1,...,n_t  
\end{equation}
\begin{equation}\label{conditions3bis}
\sum_{l=1}^{T}\mathrm{Tr}\left(i\left[ a_{ql}^{(i)}\left(a_{pl}^{(j)} \right)^{*} - a_{pl}^{(i)}\left(a_{ql}^{(j)} \right)^{*} \right]\right) = 0 ~,~ \forall q=1,...,n_t, p=q+1,...,n_t  
\end{equation}
whenever $i \in \Gamma_m$ and $j \in \Gamma_n$ and $m \neq n$, and if any group $\Gamma_i$ admits fast decodability, then the code is fast group decodable.
\end{lemma}

\begin{proof}
The proof follows from the proofs of Lemmas (\ref{groupdlema}) and (\ref{fastd}).
\end{proof}

The last class of codes studied in this work is the Block orthogonal family. We provide in the following lemma sufficient design criteria for STBCs to be Block Orthogonal of parameters $(\Gamma,k,\gamma)$.

\begin{lemma}\label{ourbo}
Let the $\mathbf{R}$ matrix of an STBC with weight matrices $\left\lbrace \mathbf{A}_1,...,\mathbf{A}_L \right\rbrace$ and $\left\lbrace \mathbf{B}_1,...,\mathbf{B}_l \right\rbrace$ be $\mathbf{R}=\left[\begin{array}{cc}
\mathbf{R}_1 & \mathbf{E} \\ \mathbf{0} & \mathbf{R}_2
\end{array} \right]$, where $\mathbf{R}_1$ is an $L \times L$ upper triangular block-orthogonal matrix with parameters $\left(\Gamma-1,k,\gamma \right)$, $\mathbf{E}$ is an $L \times l$ matrix and $\mathbf{R}_2$ is an $l \times l$ upper triangular matrix. The STBC will be block orthogonal with parameters $(\Gamma,k,\gamma)$ if the following conditions are satisfied:
\begin{itemize}
\item If there exists an ordered partition of the set of matrices $\left\lbrace \mathbf{B}_1,...,\mathbf{B}_l \right\rbrace$ into $k$ non empty subsets $S_1,...,S_k$ each of cardinality $\gamma$ such that:
\begin{equation}\label{conditions11a}
\sum_{l=1}^{T}\mathrm{Tr}\left( a_{ql}^{(i)}\left(a_{pl}^{(j)} \right)^{*} + a_{pl}^{(i)}\left(a_{ql}^{(j)} \right)^{*} \right) = 0 ~,~ \forall q=1,...,n_t, p=q,q+1,...,n_t  
\end{equation}
\begin{equation}\label{conditions22a}
\sum_{l=1}^{T}\mathrm{Tr}\left(i\left[ a_{ql}^{(i)}\left(a_{pl}^{(j)} \right)^{*} - a_{pl}^{(i)}\left(a_{ql}^{(j)} \right)^{*} \right]\right) = 0 ~,~ \forall q=1,...,n_t, p=q+1,...,n_t  
\end{equation}
 whenever $i \in S_m$ and $j \in S_n$ and $m \neq n$.
\item The matrices $\left\lbrace \mathbf{A}_1,...,\mathbf{A}_L \right\rbrace$ when used as weight matrices for an STBC yield an $\mathbf{R}$ having a block orthogonal structure with parameters $(\Gamma-1,k,\gamma)$. When $\Gamma=2$, then $L=l$ and the matrices $\left\lbrace \mathbf{A}_1,...,\mathbf{A}_L \right\rbrace$ are $k-$group decodable with variables $\gamma$ in each group.
\item The set of matrices  $\left\lbrace \mathbf{A}_1,...,\mathbf{A}_L, \mathbf{B}_1,...,\mathbf{B}_l \right\rbrace$ are such that the matrix $\mathbf{R}$ obtained is of full rank.
\item The matrix $\mathbf{E}^{t}\mathbf{E}$ is a block diagonal matrix with $k$ blocks of size $\gamma \times \gamma$.
\end{itemize}
\end{lemma}
\begin{proof}
We know from lemma (\ref{groupdlema}) that if equations (\ref{conditions11a}) and (\ref{conditions22a}) are satisfied, then the matrices $\left\lbrace \mathbf{B}_1,...,\mathbf{B}_l \right\rbrace$ are $k-$group decodable. The proof follows given the remaining design conditions are met for codes to be block-orthogonal.
\end{proof}

\paragraph*{\textbf{Example 7}: For the case of the Golden code, there are 8 LD matrices as follows:
\begin{align}
 &\mathbf{A}_1 =\mathrm{diag}\left(\nu_1,\sigma(\nu_1) \right) ,~ \mathbf{A}_2=i~\mathrm{diag}\left(\nu_1,\sigma(\nu_1) \right),~ \mathbf{A}_3= \mathrm{diag}\left(\nu_2,\sigma(\nu_2) \right),~\mathbf{A}_4=i~\mathrm{diag}\left(\nu_2,\sigma(\nu_2) \right) \notag \\
 &\mathbf{A}_5 =\mathrm{diag}\left(\nu_1,\sigma(\nu_1) \right)e ,~ \mathbf{A}_6=i~\mathrm{diag}\left(\nu_1,\sigma(\nu_1) \right)e,~ \mathbf{A}_7= \mathrm{diag}\left(\nu_2,\sigma(\nu_2) \right)e,~\mathbf{A}_8=i~\mathrm{diag}\left(\nu_2,\sigma(\nu_2) \right)e \notag
\end{align}
where $\gamma=i$ and $e=\left[ \begin{array}{cc}
0 & 1 \\
\gamma & 0
\end{array} \right]$. Now, using our criteria in theorem (\ref{ourth}), we can exactly determine the zero entries in the matrix $\mathbf{R}$. Using the conditions in equation (\ref{conditions3bis}), we have the following:
\begin{align}
\lbrace \mathbf{A}_1,\mathbf{A}_2 \rbrace:~ & l=1: \mathrm{Tr}\left(i\left[ \nu_1(0)^{*}-0(i\nu_1)^{*}\right]\right)=0 \notag \\
& l=2: \mathrm{Tr}\left(i\left[ 0(i\sigma(\nu_1))^{*}-\sigma(\nu_1)(0)^{*}\right]\right)=0  \notag
\end{align} 
\begin{align}
\lbrace \mathbf{A}_1,\mathbf{A}_4 \rbrace:~ & l=1: \mathrm{Tr}\left(i\left[ \nu_1(0)^{*}-0(i\nu_2)^{*}\right]\right)=0  \notag \\
& l=2: \mathrm{Tr}\left(i\left[ 0(i\sigma(\nu_2))^{*}-\sigma(\nu_1)(0)^{*}\right]\right)=0 \notag
\end{align}
\begin{align}
\lbrace \mathbf{A}_2,\mathbf{A}_3 \rbrace:~ & l=1: \mathrm{Tr}\left(i\left[ i\nu_1(0)^{*}-0(\nu_2)^{*}\right]\right)=0 \notag \\
& l=2: \mathrm{Tr}\left(i\left[ 0(\sigma(\nu_2))^{*}-i\sigma(\nu_1)(0)^{*}\right]\right)=0 \notag
\end{align} 
\begin{align}
\lbrace \mathbf{A}_3,\mathbf{A}_4 \rbrace:~  & l=1: \mathrm{Tr}\left(i\left[ \nu_2(0)^{*}-0(i\nu_2)^{*}\right]\right)=0 \notag \\
& l=2: \mathrm{Tr}\left(i\left[ 0(i\sigma(\nu_2))^{*}-\sigma(\nu_2)(i\sigma(\nu_2))^{*}\right]\right)=0 \notag
\end{align}
\begin{align}
\lbrace \mathbf{A}_5,\mathbf{A}_6 \rbrace:~ & l=1: \mathrm{Tr}\left(i\left[ 0(i\gamma\sigma(\nu_1))^{*}-\gamma\sigma(\nu_1)(0)^{*}\right]\right)=0 \notag \\
& l=2: \mathrm{Tr}\left(i\left[ \nu_1(0)^{*}-0(i\nu_1)^{*}\right]\right)=0 \notag
\end{align} 
\begin{align}
\lbrace \mathbf{A}_5,\mathbf{A}_8 \rbrace:~ & l=1: \mathrm{Tr}\left(i\left[ 0(i\gamma\sigma(\nu_2))^{*}-\gamma\sigma(\nu_1)(0)^{*}\right]\right)=0 \notag \\
& l=2: \mathrm{Tr}\left(i\left[ \nu_1(0)^{*}-0(i\nu_2)^{*}\right]\right)=0 \notag
\end{align}
\begin{align}
\lbrace \mathbf{A}_6,\mathbf{A}_7 \rbrace:~ & l=1: \mathrm{Tr}\left(i\left[ 0(\gamma\sigma(\nu_2))^{*}-i\gamma\sigma(\nu_1)(0)^{*}\right]\right)=0 \notag \\
& l=2: \mathrm{Tr}\left(i\left[ i\nu_1(0)^{*}-0(\nu_2)^{*}\right]\right)=0 \notag \notag
\end{align} 
\begin{align}
\lbrace \mathbf{A}_7,\mathbf{A}_8 \rbrace:~ & l=1: \mathrm{Tr}\left(i\left[ 0(i\gamma\sigma(\nu_2))^{*}-\gamma\sigma(\nu_2)(0)^{*}\right]\right)=0 \notag \\
& l=2: \mathrm{Tr}\left(i\left[ \nu_2(0)^{*}-0(i\nu_2)^{*}\right]\right)=0  \notag
\end{align}
Again for the weight matrices of the Golden code example, it is enough to check that the individual trace forms are equal to zero \cite{Mejri15} to have the sum over all columns equals $0$ to meet the conditions of equation (\ref{conditions3bis}). Given the satisfied properties, we obtain the form of the matrix $\mathbf{R}$ depicted in figure \ref{goldencmat}.
}

\section{Conclusion}
This work is dedicated to the design of STBCs in MIMO systems that admit low-complexity ML decoding using sequential decoding, for example thourgh the Sphere Decoder. We proposed novel sufficient design criteria for weight matrices defining the code for an arbitrary number of antenna and any coding rate. Our conditions explain the suboptimality of the existing approaches and show that the structure of the $\mathbf{R}$ matrix of the QR decomposition of the equivalent channel matrix depends only on the weight matrices and their ordering and not on the channel matrix or the number of antennas. Using the derived design criteria, we aim in future works to construct classes of STBCs offering low-complexity ML decoding.

\appendices
\section{Appendix: Properties of matrix $\mathbf{M}$}
\label{append1}
We show in this appendix the properties of the matrix $\mathbf{M}=\check{\mathbf{H}}^{t}\check{\mathbf{H}}$. Given the complex-to-real transformation $\check{(.)}$, column vectors of the matrix $\check{\mathbf{H}}$ are given for all $i=1,...,n_t$ by the following:
\begin{align}
&\mathbf{h}_{2i-1} = \left[ \tilde{h}_{1,2i-1},\tilde{h}_{2,2i-1},...,\tilde{h}_{n_r,2i-1} \right]^{t}= \left[\Re{(h_{1,2i-1})},\Im{(h_{1,2i-1})},..., \Re{(h_{n_r,2i-1})},\Im{(h_{n_r,2i-1})}\right]^{t}\\
&\mathbf{h}_{2i} = \left[ \bar{\bar{h}}_{1,2i-1},\bar{\bar{h}}_{2,2i-1},...,\bar{\bar{h}}_{n_r,2i-1} \right]^{t} = \left[-\Im{(h_{1,2i-1})},\Re{(h_{1,2i-1})},...,-\Im{(h_{n_r,2i-1})},\Re{(h_{n_r,2i-1})}\right]^{t}
\end{align}
Consequently we obtain:
\begin{align}
& M_{2i,2j} = <\mathbf{h}_{2i},\mathbf{h}_{2j}> = \sum_{n=1}^{n_r}\bar{\bar{h}}_{n,2i-1} \bar{\bar{h}}_{1,2j-1} = \sum_{n=1}^{n_r}\tilde{h}_{n,2i-1}\tilde{h}_{n,2j-1} = \frac{1}{2}\sum_{n=1}^{n_r}\mathrm{Tr}\left(h_{n,2i-1}h^{*}_{n,2j-1} \right), i,j=1,...,n_t \\
& M_{2i-1,2j-1} = <\mathbf{h}_{2i-1},\mathbf{h}_{2j-1}> = \sum_{n=1}^{n_r}\tilde{h}_{n,2i-1} \tilde{h}_{1,2j-1} = \frac{1}{2}\sum_{n=1}^{n_r}\mathrm{Tr}\left(h_{n,2i-1}h^{*}_{n,2j-1} \right), i,j=1,...,n_t \\
& M_{2i-1,2j} = <\mathbf{h}_{2i-1},\mathbf{h}_{2j}> = \sum_{n=1}^{n_r} \tilde{h}_{n,2i-1}\bar{\bar{h}}_{n,2j-1} = \frac{1}{2}\sum_{n=1}^{n_r}\mathrm{Tr}\left(ih_{n,2i-1}h^{*}_{n,2j-1}\right), i,j=1,...,n_t, j>i \\
& M_{2i,2j-1} = <\mathbf{h}_{2i},\mathbf{h}_{2j-1}> = \sum_{n=1}^{n_r}\bar{\bar{h}}_{n,2i-1}\tilde{h}_{n,2j-1} = \frac{1}{2}\sum_{n=1}^{n_r}\mathrm{Tr}\left(i^{*}h_{n,2i-1}h^{*}_{n,2j-1} \right),i,j=1,...,n_t, j>i
\end{align}
It follows that:
\begin{align}\label{prop}
& M_{2i-1,2i-1} = M_{2i,2i} ~,~ \forall i=1,...,n_t  \\
& M_{2i-1,2i} = 0 ~,~ \forall i=1,...,n_t  \\
& M_{2i-1,2j-1} = M_{2i,2j} ~,~ \forall 1 \leq i,j \leq n_t, j>i  \\
& M_{2i,2j-1} = - M_{2i-1,2j} ~,~  \forall 1 \leq i,j \leq n_t, j>i  
\end{align}

%\nocite{*}   
%\small{
%\bibliographystyle{plain}  
%\bibliography{midtermbib}}
\bibliographystyle{plain}  
\bibliography{JournalWCOMv3.bbl}
\end{document}